%% file: reading_saxe.tex
\newif\ifrims
\rimstrue
\rimsfalse

\ifrims
\documentclass[secnum,leqno]{rims-bessatsu}
\else
\documentclass[a4paper,11pt]{article}

\fi

\usepackage{amssymb,amsmath,amsfonts,amsopn}
\usepackage{graphics}
\usepackage[dvips]{graphicx}
\usepackage{mathrsfs}		
\usepackage{algorithmic}
\usepackage{algorithm}
\usepackage{color}
\usepackage{verbatim}
\usepackage{hyperref}
\usepackage{url}
\usepackage{eepic}
\usepackage{epic}
\usepackage{tikz}
\usetikzlibrary{graphs,quotes,positioning}
\input{definecolors.tex}

\ifrims
\newtheorem{thm}{Theorem}[section]

\newtheorem{prp}[thm]{Proposition}

\theoremstyle{definition}

\theoremstyle{remark}
\else
\usepackage{amsmath}
\usepackage{amssymb}
\usepackage{mathtools}
\usepackage{amsthm}
\newtheorem{result}{\ }[section]
\theoremstyle{changebreak}                
\newtheorem{thm}[result]{Theorem}

\newtheorem{prop}[result]{Proposition}

\fi

\newcommand{\ran}[1]{\mathsf{ran}\,{#1}}

\ifrims
\title{On Saxe's theorems about the complexity of the Distance Geometry Problem}
\author{Ma\"el \textsc{Kupperschmitt}\footnote{Ecole Polytechnique, Palaiseau, France.\newline e-mail: \texttt{mael.kupperschmitt--le-flao@polytechnique.edu}}
          ~and Leo \textsc{Liberti}\footnote{LIX CNRS, Ecole Polytechnique, Institut Polytechnique de Paris, Palaiseau, France. \endgraf e-mail: \texttt{leo.liberti@polytechnique.edu}}}
\AuthorHead{Kupperschmitt \& Liberti}         
\support{}  
\VolumeNo{x}           
\YearNo{202x}           
\PagesNo{000--000}      
\communication{Received April 20, 202x. Revised September 11, 202x.}      
\begin{document}
%
\maketitle

\else
\begin{document}

\thispagestyle{empty}
\begin{center} 

{\LARGE On Saxe's theorems about the complexity of the Distance Geometry Problem}
\par \bigskip
{\sc Ma\"el Kupperschmitt${}^{1}$, Leo Liberti${}^{2}$} 
\par \bigskip
\begin{minipage}{15cm}
\begin{flushleft}
{\small
\begin{itemize}
\item[${}^1$] {\it \'Ecole Polytechnique, Institut Polytechnique de Paris, F-91129 Palaiseau, France} \\ Email:\url{mael.kupperschmitt--le-flao@polytechnique.edu}
\item[${}^2$] {\it LIX CNRS, \'Ecole Polytechnique, Institut Polytechnique de Paris, F-91128 Palaiseau, France} \\ Email:\url{leo.liberti@polytechnique.edu}
\end{itemize}
}
\end{flushleft}
\end{minipage}
\par \medskip \today
\end{center}
\par \bigskip
\fi

\begin{abstract}      
In 1979, James B.~Saxe published an extended summary on the complexity of the Distance Geometry Problem in the proceedings of the 17th Allerton Conference \cite{saxe79}. Many of the proofs in his paper are sketches, and even the whole proofs do not have all the details. In this paper we provide a commentary to Saxe's results and hopefully more understandable versions thereof. 
\end{abstract}


\section{Introduction}
\label{s:intro}
This paper aims at clarifying the results about the computational complexity of the Euclidean Distance Geometry Problem (EDGP) given in \cite{saxe79}, which is an ``extended summary'' published in a 1979 conference proceedings volume. The full archival paper was never published after the conference. The first part of the technical report \cite{saxe80}, however, looks like it might have been a candidate to the missing full paper. Unfortunately, though, many proofs in \cite{saxe80} are as sketchy as those of the extended summary \cite{saxe79}, which also contains some graphical errors in \cite[Fig.~4.1]{saxe79} (corrected in \cite{saxe80}). 

Given $K>0$, the EDGP is the inverse problem to finding a subset of Euclidean distances (with adjacencies) given a set of points in a Euclidean space $\mathbb{R}^K$. It can also be described as drawing a weighted graph in $\mathbb{R}^K$ so that the vertices are points and the edges are segments having the same length as the edge weights. Formally, given $K>0$ and a simple undirected edge-weighted graph $G=(V,E,d)$ (where $d:E\to\mathbb{R}_+$ is the edge weight function), one must decide whether there exists a \textit{realization} $x:V\to\mathbb{R}^K$ such that:
\begin{equation}
  \forall \{u,v\}\in E \quad \|x_u - x_v\|_2^2 = d_{uv}^2. \label{dgp}
\end{equation}
Note that both sides are squared turn Eq.~\eqref{dgp} into a quadratic (polynomial) equation.

The EDGP is natively a decision problem; realizations can be used as certificates for YES instances in an appropriate computational model (see Sect.~\ref{s:compmodels}). A realization $x$ can be represented by an $n\times K$ matrix, where $n=|V|$ and the $v$-th row of $x$ is a $K$-vector representing the point position of the vertex $v$ in $\mathbb{R}^K$. If $K=1$, then a realization is simply an $n$-vector the $v$-th component of which is a scalar that indicates the position of $v$ on the real line. If $K$ is fixed, the resulting problem is denoted EDGP${}_K$. The generalization of the EDGP to other norms than $\ell_2$ is known as the Distance Geometry Problem (DGP), but most of the well-known applications use the $\ell_2$ norm \cite{dgp-sirev,dgbook,dgds}.

This paper falls in the literary category of commentaries, which are extremely common in religious studies and fairly common in the rest of the human sciences. In mathematics, this genre is typical of the writings about its history. Our aim, however, is not quite historical\footnote{For a commentary about the history of mathematics related to the DGP, see \cite{six}.}. Our motivation to re-examine the complexity of the EDGP according to \cite{saxe79} is that it supports a certain line of research we are undertaking: we want to employ complexity reductions from various problems to the EDGP in order to use its solution algorithms to solve various \textbf{NP}-hard problems \cite{dgcomp1}. 

The rest of this paper is organized as follows. We introduce preliminary notions in Sect.~\ref{s:prelim}. The content of Saxe's papers is summarized, and some critical points are discussed, in Sect.~\ref{s:critique}. In Sect.~\ref{s:innp} we provide an in-depth discussion of the membership of EDGP${}_1$ in \textbf{NP}. The main part of this paper is Sect.~\ref{s:npcomplete}, where we discuss the reduction from \textsc{3sat} to the EDGP${}_1$. The extension of this reduction to the EDGP${}_K$ is criticized in Sect.~\ref{s:edgpK}, while the extension to edge weights in $\mathbb{R}$ rather than $\mathbb{N}$ is criticized in Sect.~\ref{s:realEDGP}. In Sect.~\ref{s:ambiguous} we summarize the extension to the ``ambiguous'' EDGP variant consisting of determining whether an EDGP instance has a single or many incongruent realizations. Sect.~\ref{s:concl} concludes the paper.

\subsection{Preliminary notions}
\label{s:prelim}
We denote the range of a function $f$ by $\ran{f}$. A function $\gamma:\mathbb{R}^K\to\mathbb{R}^K$ is a \textit{congruence} if, given any $S\subset\mathbb{R}^K$, for all $x,y\in S$ we have $\|x-y\|_2=\|\gamma(x)-\gamma(y)\|_2$. Rotations, translations, and reflections are all congruences. A pair $(\psi,G)$ where $\psi:\mathbb{R}^K\to\mathbb{R}^K$ and $G=(V,E)$ with $|V|=n$ is a simple undirected graph is an \textit{isometry} if, given any finite sequence $S=(x_1,\ldots,x_n\;|\;\forall i\le n\;x_i\in\mathbb{R}^K)$ such that $|S|=|V|$, for all $\{u,v\}\in E$ we have $\|x_u-x_v\|_2=\|\psi(x_u)-\psi(x_v)\|_2$. With a small abuse of terminology we say that an isometry defined on a complete graph is a congruence. 

\subsubsection{Models of computation}
\label{s:compmodels}
For a definition of a Turing Machine (TM) see \cite{mdavis}. A \textit{nondeterministic} TM is a TM that is capable of following each branch of a test concurrently at no additional computational cost; or, equivalently, an ideal parallel computer with an unbounded number of processors. A model of computation equivalent to the TM but more similar to current computers is the Random Access Machine (RAM) \cite{minsky}, which is based on registers holding values. The real RAM \cite{shamos} is an abstract model of RAM where each register can hold a real number. See \cite{blum} for notions of computability and complexity concerning the real RAM.

\subsubsection{Computational complexity}
\label{s:compcompl}
An algorithm is thought of as a function $A$ mapping an input $\iota$ to an output. In this paper we only look at decision problems, so the output is a single bit (YES or NO). The collection of all $\iota$ that $A$ can accept as a valid input is a \textit{problem}; we say that a problem consists of \textit{instances}\footnote{Usually we assume that a problem has infinitely many instances, otherwise problems become algorithmically trivial from a theoretical point of view: one can pre-compute solutions for the constant number of instances, then provide them to users in constant time.}. The computational complexity of a terminating algorithm is the number of elementary steps, in function of the amount of memory necessary to store $\iota$, that it performs before termination. Since problems can be solved by more than a single algorithm, we define the complexity of a problem as the complexity of the best algorithm that can solve it. Problems with a polynomial complexity are known as \textit{tractable}\footnote{Polynomial complexity is invariant to many natural interpretations of a given computational model: e.g.~constant number of tapes/heads in a TM, constant number of registers in a RAM \cite{cobham}.}. The class of tractable decision problems is known as \textbf{P}.

An interesting class of decision problems which we think have computational complexity larger than any polynomial is that of Nondeterministic Polynomial-time (\textbf{NP}): its problems can be solved in polynomial-time only by a nondeterministic TM, which can be thought of as a TM with an unbounded number of heads/tapes. A decision problem $P$ is in \textbf{NP} if there is a polynomial-time algorithm for a nondeterministic TM that solves $P$. An alternative definition of \textbf{NP} is that each of its problems $P$ have two properties: (i) all YES instances $p\in P$ come with a \textit{certificate} $c$ that proves that $p$ is a YES instance; (ii) $P$ comes with a polynomial-time \textit{verifier} algorithm $V_P(p,c)$ that outputs $\mathsf{TRUE}$ if the $c$ indeed proves that $p$ is YES. Many interesting problems fall in the class \textbf{NP}, and for many of which we do not know of any polynomial-time algorithm that solve them.

To characterize these problems as ``difficult'' (as opposed to the tractable ones in \textbf{P}) we use the notion of \textit{hardness} for the class \textbf{NP}: a problem $P\in\mathbf{NP}$ is \textbf{NP}-hard if, for all problems $Q\in\mathbf{NP}$, one can find a polynomial-time algorithm $\rho_{QP}$ that transforms each instance $q\in Q$ into an instance $p\in P$ such that $p$ is YES if and only if $q$ is YES\footnote{$Q$ is known as the \textit{source} problem and $P$ as the \textit{target} problem.}. If we assume that $Q$ is more difficult than $P$, then we could apply $\rho_{QP}$ to each instance of $Q$ and obtain a corresponding instance of $P$ together with the YES/NO answer using the easier algorithm for $P$, which contradicts the assumption. Therefore $Q$ can be at worst just as difficult as $P$. From this qualitative point of view, $P$ must then be hardest with respect to the class \textbf{NP}. We know that there are many \textbf{NP}-hard problems, and we also know that there are \textbf{NP}-hard problems that do not even belong to \textbf{NP}. Those that do are known as \textbf{NP}-complete.

A \textit{pseudo-polynomial-time} algorithm is polynomial in the values of the numbers of its input rather than in the amount of storage necessary to encode them (the input length). E.g.~if the input consists of $n$ integers each with maximum value $M$, then an $O(nM)$ algorithm is pseudo-polynomial-time. Since the input can be stored in $O(n\log_2 M)$ bits, the running time of a pseudo-polynomial-time algorithm could be exponential in the input length. If there is a pseudo-polynomial-time algorithm for solving $Q$, then $P$ is \textit{weakly} \textbf{NP}-hard, otherwise it is \textit{strongly} \textbf{NP}-hard.

See \cite{papacomplexity} for more information about computational complexity.

\subsubsection{Satisfiability}
\label{s:sat}
The satisfiability (\textsc{sat}) problem asks to find an assignment of TRUE/FALSE boolean values to \textit{logical variables} $s_1,\ldots,s_n$ that must satisfy a given logical expression involving AND, OR, NOT operations. In this paper we assume that the expression is in conjunctive normal form (CNF), i.e.~it consists of a conjunction of \textit{clauses}, each of which consists of a disjunction of \textit{literals}, each of which is either a variable $s_j$ or its negated form $\bar{s}_j$ (where $j\le n$). A \textsc{sat} instance therefore looks like:
\begin{equation}
  \bigwedge_{i\le m} \big(\bigvee_{j\in C^+_i} s_j \vee \bigvee_{j\in C^-_i} \bar{s}_j\big),
\end{equation}
where $C^+_i$ is an index set of variables and $C^-_i$ is an index set of negated variables occurring in the $i$-th clause. The $k$\textsc{sat} problem limits the number of literals in each clause to at most $k$ (or exactly $k$: the two requirements are equivalent because repetition of a literal in a clause does not change the truth value of the expression). Thus, \textsc{3sat} has at most three literals per clause. The \textsc{sat} and \textsc{3sat} problems are often used as source problems in polynomial reductions to prove \textbf{NP}-hardness. We recall that $k$\textsc{sat} is in \textbf{P} for $k\le 2$ \cite{papacomplexity} and \textbf{NP}-complete \cite{cook} for $k\ge 3$.

\subsubsection{Graphs and their rigidity}
Given a graph $G=(V,E)$ and a subgraph $H$ of $G$ we let $V(H)$ be the set of vertices of $H$ and $E(H)$ the set of edges of $H$; for a subset of vertices $U\subseteq V$ we define $G[U]$ to be the subgraph induced by $U$, i.e.~the graph $(U,F)$ where $F\subseteq E$ such that for all $u,v\in U$ we have $\{u,v\}\in F$ if $\{u,v\}\in E$.

An EDGP instance $(K,G)$ has either zero or uncountably many realizations. This occurs because, if there is a realization $x$ of $G$ in $\mathbb{R}^K$, then every translation, reflection (for any $K>0$) and rotation (for $K>1$) of $x$ is another realization. If we define an equivalence relation $x \approx y$ based on $y$ being a congruent image of $x$, and we consider the set of realizations of $G$ in $\mathbb{R}^K$ modulo $\approx$, then the EDGP instance may have zero, finitely many, or uncountably many realizations \cite[\S 3.2]{dgds}. For YES instances of the EDGP, given a realization $x$ of $G$ in $\mathbb{R}^K$, the pair $(G,x)$ is a \textit{framework}. The framework is \textit{rigid} if there are finitely many realizations modulo congruences, and \textit{flexible} if there are uncountably many. If a rigid framework has a unique realization, it is \textit{globally rigid}. EDGP instances leading to globally rigid frameworks are highly desirable in applications where one tries to reconstruct the position of vertices from a sparse set of distances, since there is usually a single reality that gave rise to the distance measurements\footnote{There are exceptions to the desirability of globally rigid framework, for example the reconstruction of molecule shapes from distance measurements is naturally ambiguous because of isomers: different isomers have different properties, and it may be useful to find all possible isometric but incongruent frameworks of a given EDGP.}.

\section{Contents of Saxe's paper \cite{saxe79}}
\label{s:critique}
Saxe's 1979 extended summary is titled ``Embeddability of Weighted Graphs in k-Space is Strongly NP-Hard''. Today, the term ``embeddability'' is perceived to have a wider scope than Eq.~\eqref{dgp}, because graphs may be embedded on manifolds with various topologies \cite{lando}. Saxe motivates his study of this problem with the application to distributed sensor networks, in the same spirit as \cite{yemini78}. Wireless sensors may communicate in a peer-to-peer fashion with other nearby sensors, and estimate the distance by means of the quantity of battery used for the exchange: this yields an edge-weighted graph (the input of the EDGP) in two or three dimensions.

The given EDGP instance $(K,G)$ is rational: $K\in\mathbb{N}$, $V\subset\mathbb{N}$, $E$ can be encoded by a boolean matrix, and $d$ maps edges to their weights in $\mathbb{Q}$. These rational weights can be reduced to integer by scaling by their minimum common multiple. We assume that $G$ is connected, otherwise every connected component can be realized independently. Saxe shows that that EDGP${}_1$ is in \textbf{NP} by giving a natural, direct solution algorithm for a nondeterministic TM. 

The main contributions of \cite{saxe79} are: (i) the weak \textbf{NP}-hardness of the EDGP${}_1$; (ii) the strong \textbf{NP}-hardness of the EDGP${}_1$; (iii) an extension of (ii) to the EDGP${}_K$ for any $K>0$; (iv) an extension of (ii) to EDGP instances with real (rather than rational) input, for which approximate realizations are accepted as a YES certificates; (v) the strong \textbf{NP}-hardness of the ``ambiguous problem'' of the EDGP (let us call this A-EDGP), i.e.~that of determining whether, given an EDGP instance $(K,G)$ and a realization $x\in\mathbb{R}^K$, there exists another realization $y$, incongruent to $x$, that satisfies Eq.~\eqref{dgp} for the instance $(K,G)$. 


The weak \textbf{NP}-hardness of EDGP${}_1$ was independently discovered by reduction from \textsc{partition} by Yemini\footnote{In his presentation, Yemini refers to the EDGP${}_2$, although in his reduction from \textsc{partition} the second dimension is irrelevant. Although formally one can infer (by inclusion) the weak \textbf{NP}-hardness of EDGP${}_2$ from Yemini's proof, in practice the proof really only works for the EDGP${}_1$.} \cite{yemini}, who called it the ``reflection problem'' (applicable to a certain subclass of instances of EDGP${}_2$, which is itself equivalent to EDGP${}_1$). Both authors acknowledge each other's independent discovery. The proof of this result is in fact self-contained, but its compactness obscures its meaning. A didactical exposition of this proof can be found in \cite[\S 2.4.2]{dgbook}; therefore we only recall the essence of this result. A \textsc{partition} instance on $n$ integers consists of $n$ integers $a_1,\ldots,a_n$; an instance is YES iff there is a subset $I\subseteq\{1,\ldots,n\}$ such that $\sum_{i\in I} a_i=\sum_{i\not\in I} a_i$. The instance is reduced to a single-cycle graph on vertices $V=\{1,\ldots,n+1\}$ and edges $\{v,v+1\}$ (for $1\le v\le n$), each with weight $a_v$. YES instances of \textsc{partition} reduce to single cycle graphs that can be realized so that $x_{n+1}=x_1$ (which closes the cycle). Vertices are realized so that $x_{v+1}$ is on the right (resp.~left) of $x_v$ if $a_v$ is in the left (resp.~right) hand side sum $\sum_{v\in I} a_v=\sum_{v\not\in I} a_v$ of the \textsc{partition} certificate. Since \textsc{partition} is weakly \textbf{NP}-complete, this proof only shows the weak \textbf{NP}-completeness of EDGP${}_1$.

The strong \textbf{NP}-hardness proof for EDGP${}_1$ is by reduction from \textsc{3sat}: each clause is represented by a gadget\footnote{In complexity reductions from a problem $P$ to another problem $Q$, sometimes the basic entities of $P$ are reduced to more or less complicated relations of entities in terms of $Q$, that are called ``gadgets''; if $Q$ is expressed using a graph, then the graph might be the resulting graph may be formed by subgraphs each of which is a gadget graph representing an entity in $P$.} graph, having all edge weights in the set $\{1,2,3,4\}$, on three literals that are related by another gadget (having edge weights in the set $\{1,2\}$). The whole graph is fixed with respect to translations and reflections by anchoring (i.e., fixing) two of its vertices to given positions. Both in \cite{saxe79,saxe80} the proof essentially states that the reader should study the gadgets carefully (the most complicated of which is wrong in \cite{saxe79} but corrected in \cite{saxe80}) to convince herself of its validity. In this paper we provide an explicit proof of this ``careful study'' step. Saxe also offers a second version of the clause gadget that decreases the set of edge weights of the clause gadgets to just $\{1,2\}$, by replacing the edges with length $3$ and $4$ by further graph gadgets.

An inductive gadget construction and simplification shows, basically by the same reduction from \textsc{3sat}, that EDGP${}_K$ is \textbf{NP}-hard for all $K>0$. If one understands the reduction from \textsc{3sat} to EDGP${}_1$, it is intuitively (albeit not perhaps formally) clear that this inductive construction works. 

Next, Saxe introduces the notion of approximate realizability in graphs $G$ where some of the edge weights are irrational. First, for a given $0<\varepsilon<1$, a mapping $x:V\to\mathbb{R}^K$ is a \textit{$\varepsilon$-approximate realization} if:
\begin{equation}
  \forall \{u,v\}\in E\quad (1-\alpha) d_{uv} \le \|x_u-x_v\|_2 \le (1+\alpha) d_{uv}. \label{approxdgp}
\end{equation}
Based on Eq.~\eqref{approxdgp}, Saxe defines an $(\epsilon,\delta)$-approximate EDGP is as follows: given $0\le\epsilon<\delta<1$ and an EDGP instance $(K,G)$, determine whether there is a $\delta$-approximate realization, or there is no $\epsilon$-approximate realization. Saxe notes that this problem excludes instances for which neither condition holds, and suggests that such instances can be declared YES or NO arbitrarily: Saxe concedes that the problem is somewhat ill-defined. By means of his \textsc{3sat} to EDGP${}_1$ reduction, Saxe proves that the $(1/18,1/9)$-approximate EDGP is \textbf{NP}-hard, by noting that the gadgets in the version of the proof where edge weights are in $\{1,2\}$ there are cycles of length at most $16$: the proof is not clearly understandable. This result is generalized to a more helpful $(\delta,\epsilon)$-approximate EDGP being \textbf{NP}-hard for all $0\le\delta\le\epsilon<1/8$ in \cite{saxe80}. Saxe also notes that the weak reduction from \textsc{partition} to a single cycle graph does not lend itself to this analysis since there are FPTAS for \textsc{partition} \cite{ibarra}, which would make it impossible to claim that approximate realizations are hard to find.

The proof of the \textbf{NP}-completeness of the A-EDGP${}_1$ is sketched in \cite{saxe79} as follows: (1) define ambiguous problems related to \textsc{3sat} and \textsc{4sat} similarly to the A-EDGP${}_1$ (let us call them, respectively, \textsc{a-3sat} and \textsc{a-4sat}); (2) \textbf{NP}-completeness of the \textsc{a-4sat} follows by reduction from \textsc{3sat}; (3) \textbf{NP}-completeness of \textsc{a-3sat} is by reduction from \textsc{a-4sat}; (4) the strong \textbf{NP}-completeness of A-EDGP${}_1$ is by reduction from \textsc{a-3sat}, based on a uniqueness property of the \textsc{3sat} to EDGP${}_1$ reduction; (5) for any $K>1$, the strong \textbf{NP}-hardness of A-EDGP${}_K$ follows by reduction from A-EDGP${}_{K-1}$. We think that this is the only result of \cite{saxe79} that is actually clarified and expanded in \cite{saxe80}.


\section{Membership of EDGP${}_1$ in \textbf{NP}}
\label{s:innp}
As is well known, there are two equivalent definitions for a problem $P$ to belong to the class \textbf{NP} \cite{papacomplexity}: the first is that there must exists a polynomial-time algorithm running on a nondeterministic TM that solves all instances of $P$. The second is that every YES instance should possess a certificate that can be verified in polynomial time by a deterministic TM.

Saxe's proof that EDGP${}_1$ is in \textbf{NP} is based on the first definition. We provide a clearer proof here.
\ifrims\begin{prp}\else\begin{prop}\fi
  There is a nondeterministic polynomial-time TM for solving EDGP${}_1$.
  \label{prop:ndtm}
\ifrims\end{prp}\else\end{prop}\fi
\begin{proof}
The algorithm computes a realization $x_v$ for each $v\in V$, where $G=(V,E)$ is assume connected (otherwise the algorithm is deployed on each connected component). The algorithm proceeds with a depth-first search (DFS) from an arbitrary starting vertex $z$ placed at $x_z=0$. Then each edge $\{z,v\}$ in the star of $z$ is considered, and $v$ is placed at a position $x_v=x_z\pm d_{zv}$ (each $\pm$ alternative is carried out concurrently since our TM is nondeterministic). Suppose $u$ is inductively placed at $x_u$ in such a way that the set $U$ of vertices already considered before $u$ define an induced subgraph $H=G[U]$ of $G$ that is consistently realized in $\mathbb{R}$ (i.e.~all edges are realized as segments of length equal to their weights). Then for each edge $\{u,v\}$ in the star of $u$, either $v\in V(H)$ or $v\not\in V(H)$. If $v\in V(H)$ then either $|x_u-x_v|=d_{uv}$ or not. If yes, the computation continues, otherwise this particular computation branch stops (recall the TM is nondeterministic). If $v\not\in V(H)$ then $x_v=x_u\pm d_{uv}$, where the two alternatives are carried out concurrently. If at least one computation branch visits all vertices, the the EDGP${}_1$ instance is YES and the realization $x$ is a certificate. Conversely, if all computation branches terminate prematurely, the instance is NO. This algorithm takes $O(n+m)$ on a nondeterministic TM. 
\end{proof}
We remark that, if the algorithm in Prop.~\ref{prop:ndtm} is instead deployed on a \textit{deterministic} TM, it leads to a worst-case exponential-time branching algorithm for the EDGP that is called Branch-and-Prune (BP) \cite{lln5}. 

A proof based on the other definition of \textbf{NP}  must address a numerical difficulty given by the fact that, when using the realization of a YES instance as a certificate, the realization components may not be rational even if the instance is rational, at least for $K>1$ (consider for example an equilateral triangle in 2D). Therefore we must show that, for $K=1$, there is always a rational realization for every YES instance. Let $x$ be a realization of the connected EDGP${}_1$ instance graph $G$, and assume $x$ has at least one irrational component. Since all distances are rational and the graph is connected, every other component must also be rational, otherwise there would be at least one irrational distance, which is not the case. Then if one irrational component is translated to any rational number, all other components will also be translated to rational numbers, again because all distances are rational.

We remark that, because Eq.~\eqref{dgp} is a system of (quadratic) polynomial equations, all realizations of EDGP instances must only involve algebraic components, which are finitely representable. So there may be a chance that the EDGP${}_K$ is in \textbf{NP} even for $K>1$, but this is unknown. Some basic attempts in this sense have failed \cite{dgpinnp}. 

\section{EDGP${}_1$ is strongly \textbf{NP}-hard}
\label{s:npcomplete}
In the reduction from \textsc{3sat} to EDGP${}_1$, Saxe presents three graph gadgets, two of which are simpler (for variables and literals), and one is more complicated (for clauses). We think that it makes more sense to merge the two simpler ones, see Fig.~\ref{fig:saxe12}. 
\begin{figure}[!ht]
  \begin{center}
\begin{tikzpicture}[
    vertex/.style={circle,draw,fill=blue!20,minimum size=15pt,inner sep=0pt},
    edge/.style={thick},
    scale=1.2]
  \node[vertex](A) at (1.5,4) {$\mathsf{A}$};
  \node[vertex](B) at (4.5,4) {$\mathsf{B}$};
  \node[vertex] (u1) at (0, 3) {$s_1$};
  \node[vertex] (u2) at (0, 2) {$s_2$};
  \node[vertex] (u3) at (0, 0.5) {$s_n$};  
  \node[vertex] (v1) at (3, 3) {$\bar{s}_1$};
  \node[vertex] (v2) at (3, 2) {$\bar{s}_2$};
  \node[vertex] (v3) at (3, 0.5) {$\bar{s}_n$};
  \draw[edge] (u1) -- (v1) node[midway, above] {2};
  \draw[edge] (u2) -- (v2) node[midway, above] {2};
  \draw[edge] (u3) -- (v3) node[midway, above] {2};
  \draw[edge] (A) -- (u1) node[midway, above] {1};
  \draw[edge] (A) -- (u2) node[above, pos=0.9] {1};
  \draw[edge] (A) -- (u3) node[midway, below, pos=0.7] {1};
  \draw[edge] (A) -- (v1) node[midway, above] {1};
  \draw[edge] (A) -- (v2) node[above, pos=0.9] {1};
  \draw[edge] (A) -- (v3) node[midway, below, pos=0.7] {1};
  \draw[edge] (A) -- (B) node[midway, above] {2};
  \node at (0, 1.3) {$\vdots$};
  \node at (3, 1.3) {$\vdots$};
\end{tikzpicture}
\end{center}
\caption{The gadget for literals used in the \textsc{3sat}$\to$EDGP${}_1$ reduction: the \textsc{3sat} variables appear in vertices as either $s_j$ or $\bar{s}_j$ for $j\le n$.}
\label{fig:saxe12}
\end{figure}
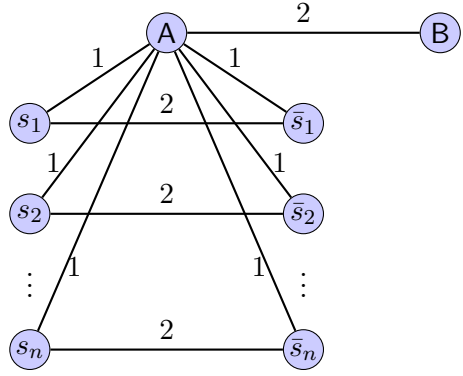
As can be seen from the figure, this gadget introduces two arbitrary vertices $\mathsf{A},\mathsf{B}$ (having fixed realizations at $0$ and $2$ respectively) linked by an edge with weight $2$, and assigns each literal to a vertex. There are edges (with weight equal to $2$) between the vertices for $s_j$ and $\bar{s}_j$, there are edges (with unit weight) between each literal vertex and $\mathsf{A}$.

The gadget graph for \textsc{3sat} clauses is shown in Fig.~\ref{fig:saxe3}. It differs from the gadget shown in \cite{saxe79} (which contains some errors), but it is consistent with the one shown in \cite{saxe80}. 
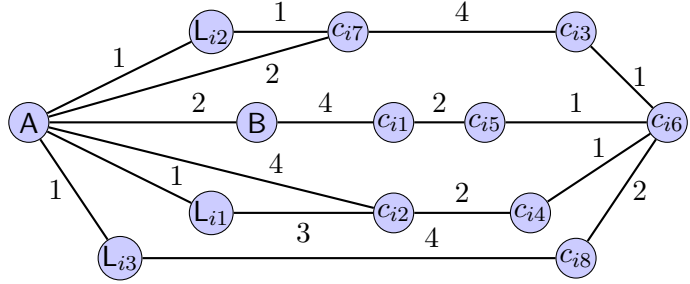
\begin{figure}[!ht]
  \begin{center}
\begin{tikzpicture}[
    vertex/.style={circle,draw,fill=blue!20,minimum size=15pt,inner sep=0pt},
    edge/.style={thick},
    scale=1.2]
  \node[vertex] (A) at (1,4) {$\mathsf{A}$};
  \node[vertex] (B) at (3.5,4) {$\mathsf{B}$};    
  \node[vertex] (c1) at (5,4) {$c_{i1}$}; 
  \node[vertex] (c2) at (5,3) {$c_{i2}$}; 
  \node[vertex] (c3) at (7,5) {$c_{i3}$}; 
  \node[vertex] (c4) at (6.5,3) {$c_{i4}$}; 
  \node[vertex] (c5) at (6,4) {$c_{i5}$}; 
  \node[vertex] (c6) at (8,4) {$c_{i6}$}; 
  \node[vertex] (c7) at (4.5,5) {$c_{i7}$}; 
  \node[vertex] (c8) at (7,2.5) {$c_{i8}$}; 
  \node[vertex] (L1) at (3,3) {$\mathsf{L}_{i1}$}; 
  \node[vertex] (L2) at (3,5) {$\mathsf{L}_{i2}$}; 
  \node[vertex] (L3) at (2,2.5) {$\mathsf{L}_{i3}$}; 
  \draw[edge] (A) -- (B) node[midway,above,pos=0.8] {2};
  \draw[edge] (A) -- (L1) node[midway,above,pos=0.9] {1};
  \draw[edge] (A) -- (c2) node[midway,above,pos=0.7] {4};  
  \draw[edge] (A) -- (L2) node[midway,above] {1};
  \draw[edge] (A) -- (c7) node[midway,below,pos=0.8] {2};  
  \draw[edge] (A) -- (L3) node[midway,left] {1};
  \draw[edge] (L1) -- (c2) node[midway,below] {3};
  \draw[edge] (L2) -- (c7) node[midway,above] {1};
  \draw[edge] (c2) -- (c4) node[midway,above] {2};
  \draw[edge] (c7) -- (c3) node[midway,above] {4};
  \draw[edge] (B) -- (c1) node[midway,above] {4};
  \draw[edge] (c1) -- (c5) node[midway,above] {2};
  \draw[edge] (c5) -- (c6) node[midway,above] {1};
  \draw[edge] (c3) -- (c6) node[midway,right] {1};
  \draw[edge] (c4) -- (c6) node[midway,above] {1};
  \draw[edge] (L3) -- (c8) node[midway,above,pos=0.7] {4};
  \draw[edge] (c6) -- (c8) node[midway,right] {2};
\end{tikzpicture}
\end{center}
\caption{The gadget for clause $i$ used in the \textsc{3sat}$\to$EDGP${}_1$ reduction. There are eight additional vertices for each clause $i\le m$, labelled $c_{i1},\ldots,c_{i8}$. The vertices labeled $\mathsf{L}_{ih}$ represent the $h$-th (positive or negative) literal in clause $i$ for each $h\le 3$, since each clause has exactly three literals.}
\label{fig:saxe3}
\end{figure}

A few remarks are in order.
\begin{itemize}
\item Each clause gadget includes the vertices $\mathsf{A},\mathsf{B}$, which are the same vertices appearing in the gadget for literals: therefore the whole graph, consisting of the union of the literals gadget and the $m$ clause gadgets, is connected. Note that $\mathsf{A},\mathsf{B}$ are \textit{anchors}, i.e.~vertices where the position is given: $x_{\mathsf{A}}=0$ and $x_{\mathsf{B}}=2$.
\item All of the literal vertices introduced in the literals gadget appear in the $i$-th clause gadget whenever the literal appears in the $i$-th clause of the \textsc{3sat} instance: the vertex labelled $\mathsf{L}_{ih}$ (for $h\in\{1,2,3\}$) is the same vertex appearing in Fig.~\ref{fig:saxe12} whenever the literal $\ell_j$ (which may be a positive or negative logical variable) is the $h$-th literal appearing in the $i$-th clause. Thus, the connectivity of the clause gadget graphs is not limited to the vertices $\mathsf{A},\mathsf{B}$. The additional vertices $c_{i1},\ldots,c_{i8}$ define various paths and cycles that guarantee that, for each $i\le m$, the $i$-th clause graph can be realized in $\mathbb{R}$ iff the $i$-th \textsc{3sat} clause is satisfied (this will be proved later).
\item It is \textit{not} the variables that are represented as vertices, but the literals. We will therefore not read the value of $s_j$ from its realized position; rather, the realized positions of the literals corresponding to the vertices labelled $s_j,\bar{s}_j$ will determine the value of the logical variable $s_j$ (for $j\le n$). More precisely, if the vertex labelled $s_j$ is realized at $x_{s_j}=1$, then because of the fixed position of $\mathsf{A},\mathsf{B}$, the vertex labelled $\bar{s}_j$ is realized at $x_{\bar{s}_j}=-1$, which imply that the logical variable $s_j$ will have value TRUE. If, instead, $x_{s_j}=-1$ (implying that $x_{\bar{s}_j}=1$), then $s_j$ will have value FALSE.
\end{itemize}

Let $\mathcal{G}$ be the graph obtained by the union of the graph gadgets in Fig.~\ref{fig:saxe12}-\ref{fig:saxe3}. 
\begin{thm}
  \label{thm:3sat2edgp1}
The \textsc{3sat} instance is YES iff the EDGP${}_1$ instance $\mathcal{G}$ is YES.
\end{thm}
The text provided in both \cite{saxe79,saxe80} as proof of Thm.~\ref{thm:3sat2edgp1} is simply
\begin{quote}
  {\small ``Careful study of the graph [\dots] will reveal that it is impossible to embed it in the line in such a way that $\mathsf{A}$ is sent to $0$, $\mathsf{B}$ is sent to $2$, and all three of the $\mathsf{L}_{ih}$ are sent to $-1$ (\textsf{FALSE}), but if one or more of the $\mathsf{L}_{ih}$ are sent to $1$ (\textsf{TRUE}), then an embedding is possible.''\par}
\end{quote}
  Since this proof is crucial for most of the subsequent results of \cite{saxe79}, we provide a full proof. 
\begin{proof}
  First, we assume that the \textsc{3sat} instance is NO and suppose, to aim at a contradiction, that the EDGP${}_1$ instance has a realization. Since the \textsc{3sat} instance is NO, there is no satisfying assignment for all clauses: hence at least one clause, say the $i$-th clause, must be such that all its literals are realized at $-1$. Now consider the four paths $P_1 = (\mathsf{A},\mathsf{L}_{i2},c_{i7},c_{i3},c_{i6})$, $P_2 = (\mathsf{A},\mathsf{B},c_{i1},c_{i5},c_{i6})$, $P_3 = (\mathsf{A},\mathsf{L}_{i1},c_{i2},c_{i4},c_{i6})$ and $P_4 = (\mathsf{A},\mathsf{L}_{i3},c_{i8},c_{i6})$ in $\mathcal{G}$: note that they all end at the vertex $c_{i6}$. As we supposed that $\mathcal{G}$ is realizable, regardless of the path we take, their final vertex $c_{i6}$ must be at a unique realized position. Path $P_3$ and the edge weight of $\{\mathsf{A},c_{i2}\}$ being $2$ forces the position of $c_{i2}$ at $-4$; path $P_1$ and edge weight of $\{\mathsf{A},c_{i7}\}$ being $4$ forces the position of $c_{i7}$ at $-2$. Next, we use the edge weights of the four paths and the fixed positions of the anchors $\mathsf{A},\mathsf{B}$ to compute the possible positions of $c_{i6}$. $P_1$ allows the position of $c_{i6}$ to be at $-1-1\pm 4\pm 1$, that is at any value in the set $S_1=\{-7,-5,1,3\}$. $P_2$ allows the position of $c_{i6}$ to be at $0+2\pm 4\pm 2\pm 1$, that is at any value in the set $S_2=\{\pm 5,\pm 3,\pm 1,7,9\}$. $P_3$ allows the position of $c_{i6}$ to be at $-1-3\pm 2\pm 1$, that is at any value in the set $S_3=\{-7,-5,-3,-1\}$. $P_4$ allows the position of $c_{i6}$ to be at $-1\pm 4\pm 2$, that is at any value in the set $S_4=\{-7,-3,1,5\}$. Hence the position of $c_{i6}$ must be in the intersection $\bigcap_{1\le h\le 4} S_h$ which is, by inspection, empty. This negates the realizability of $\mathcal{G}$. 

  Next, we assume that the \textsc{3sat} instance is YES, which implies that no triplet of literals occurring in any clause is ever assigned the values (FALSE, FALSE, FALSE). This, in turn, implies that, for any clause $i\le m$, the vertex triplet $(\mathsf{L}_{i1}, \mathsf{L}_{i2}, \mathsf{L}_{i3})$ can never be realized at $(-1,-1,-1)$. By straightforward computation we show, in the table below, that all of the gadget vertices can be realized.
  \begin{center}
    \begin{tabular}{ll}
      $(\mathsf{L}_{i1}, \mathsf{L}_{i2}, \mathsf{L}_{i3})$ & $(c_{i1},c_{i2},c_{i3},c_{i4},c_{i5},c_{i6},c_{i7},c_{i8})$ \\ \hline
      $(1,1,1)$ & $(6,4,6,6,8,2,7,5)$ \\
      $(1,1,-1)$ & $(6,4,6,6,4,5,2,3)$  \\
      $(1,-1,1)$ & $(6,4,4,2,4,6,-2,5)$ \\
      $(-1,1,1)$ & $(-2,-4,-2,-2,0,-1,2,-3)$ \\
      $(1,-1,-1)$ & $(-2,4,2,2,0,1,-2,3)$ \\
      $(-1,1,-1)$ & $(-2,-4,-4,-2,-4,-3,2,-5)$ \\
      $(-1,-1,1)$ & $(-2,-4,-6,-6,-4,-5,-2,-3)$ 
    \end{tabular}
  \end{center}
  This concludes the proof.
\end{proof}
There is one last detail to settle in this reduction, which will be important for the A-EDGP. After the citation from \cite{saxe79} quoted above, Saxe adds that, if a realization of $\mathcal{G}$ is possible, then ``in fact, exactly one such embedding is possible''. This should be interpreted to mean that, if the \textsc{3sat} instance is YES, for every satisfiable assignment of boolean values to \textsc{3sat} variables there is a corresponding realization of the graph $\mathcal{G}$ (the precise positions of vertices in each clause gadget are given in the last part of the proof of Thm.~\ref{thm:3sat2edgp1}). 

At this point, Saxe notes that the edges having weights $3$ and $4$ in $\mathcal{G}$ can be replaced by further gadgets (called $T_3, T_4$ in \cite[Fig.~4.2]{saxe79}) only involving edges with weights $1$ and $2$, which yields a reduced EDGP${}_1$ instance graph $\mathcal{G}_1$ that only uses edges with weights in $\{1,2\}$. 
\begin{figure}[!ht]
  \begin{center}
\begin{tikzpicture}[
    vertex/.style={circle,draw,fill=blue!20,minimum size=15pt,inner sep=0pt},
    edge/.style={thick},
    scale=1.5]
  \node[vertex] (T31) at (0,0) {${e_1}$};
  \node[vertex] (T32) at (2,0) {${e_3}$};    
  \node[vertex] (T33) at (1,1) {${e_2}$}; 
  \node[vertex] (T34) at (3,1) {${e_4}$}; 
  \draw[edge] (T31) -- (T32) node[midway,above] {2};
  \draw[edge] (T31) -- (T33) node[midway,above] {1};
  \draw[edge] (T32) -- (T33) node[midway,above] {1};  
  \draw[edge] (T32) -- (T34) node[midway,above] {1};
  \draw[edge] (T33) -- (T34) node[midway,below] {2};  
\end{tikzpicture}
\begin{tikzpicture}[
    vertex/.style={circle,draw,fill=blue!20,minimum size=15pt,inner sep=0pt},
    edge/.style={thick},
    scale=1.5]
  \node[vertex] (T31) at (0,0) {${e_1}$};
  \node[vertex] (T32) at (2,0) {${e_3}$};    
  \node[vertex] (T33) at (1,1) {${e_2}$}; 
  \node[vertex] (T34) at (3,1) {${e_4}$};
  \node[vertex] (T35) at (4,0) {${e_5}$};  
  \draw[edge] (T31) -- (T32) node[midway,above] {2};
  \draw[edge] (T31) -- (T33) node[midway,above] {1};
  \draw[edge] (T32) -- (T33) node[midway,above] {1};  
  \draw[edge] (T32) -- (T34) node[midway,above] {1};
  \draw[edge] (T32) -- (T35) node[midway,above] {2};
  \draw[edge] (T33) -- (T34) node[midway,below] {2};
  \draw[edge] (T34) -- (T35) node[midway,below] {1};
  
\end{tikzpicture}
\end{center}
\caption{The gadget graphs $T_3$ (left) and $T_4$ (right) show that $\mathcal{G}$ is a YES instance of EDGP${}_1$ iff $\mathcal{G}$ modified so that every edge of weight $3$ (resp.~$4$) is replaced by the gadget $T_3$ (resp.~$T_4$).}
\label{fig:saxeT}
\end{figure}
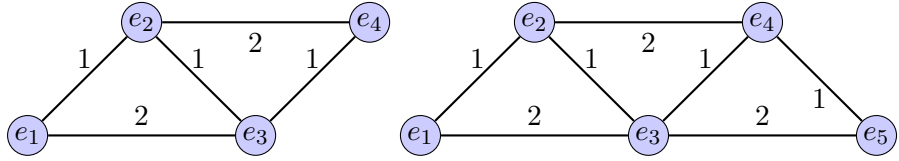
No formal proof is offered for this fact in \cite{saxe79}, but the gadgets are very simple (see Fig.~\ref{fig:saxeT}) and they clearly achieve what they are designed for: $T_3$ can only be realized in the real line as a segment of length $3$, and $T_4$ can only be realized in the real line as a segment of length $4$, exactly like edges with weights $3$ and $4$ would. Specifically, if the edge $\{u,v\}$ is to be replaced by the gadget $T_h$ for $h\in\{3,4\}$, $u\equiv e_1$ and $v\equiv e_{h+1}$ in the replacement graph. This way, $\mathcal{G}$ is a YES instance of EDGP${}_1$ iff $\mathcal{G}_1$ is a YES instance of EDGP${}_1$.

\section{EDGP${}_K$ is strongly \textbf{NP}-hard for any $K\ge 1$}
\label{s:edgpK}
In order to prove the \textbf{NP}-hardness of EDGP${}_K$ with $K>1$, Saxe proceeds by induction. He designs two further gadget graphs $R_1,R_2$ \cite[Fig.~4.3]{saxe79} based on the Pythagorical triplet $(3,4,5)$ as shown in Fig.~\ref{fig:saxeR}. It is claimed that the graph $\mathcal{G}_1$ (with weights in $\{1,2\}$ is a YES instance of EDGP${}_1$ iff the graph $\mathcal{G}'$ obtained by replacing all edges $\{u,v\}$ with weight $1$ with $R_1$ (setting $u\equiv f_1$, $v\equiv f_3$ in the replacement) and all edges $\{u,v\}$ with weight $2$ with $R_2$ (setting $u\equiv f_1$, $v\equiv f_5$ in the replacement) is a YES instance of EDGP${}_2$. Then he states that $\mathcal{G}'$ can be reduced to a graph $\mathcal{G}_2$ with weights in $\{1,2\}$ by using $T_3,T_4$ and similar gadgets for edges with weights in $\{5,8\}$: this starts the induction. Finally he states that, by induction, $\mathcal{G}_{K-1}$ is a YES instance of the EDGP${}_{K-1}$ iff $\mathcal{G}_{K}$ is a YES instance of EDGP${}_K$. 
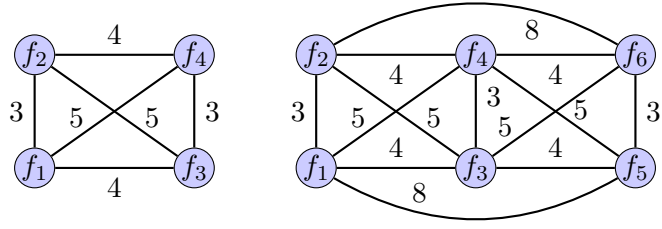
\begin{figure}[!ht]
  \begin{center}
    \begin{tabular}{lr}
      \makebox{
        \begin{tikzpicture}[baseline=(current bounding box.center),
    vertex/.style={circle,draw,fill=blue!20,minimum size=15pt,inner sep=0pt},
    edge/.style={thick},
    scale=1.5]
  \node[vertex] (R11) at (0,0) {${f_1}$};
  \node[vertex] (R12) at (1.4,0) {${f_3}$};    
  \node[vertex] (R13) at (1.4,1) {${f_4}$}; 
  \node[vertex] (R14) at (0,1) {${f_2}$}; 
  \draw[edge] (R11) -- (R12) node[midway,below] {4};
  \draw[edge] (R11) -- (R13) node[midway,above,pos=0.2] {5};
  \draw[edge] (R11) -- (R14) node[midway,left] {3};  
  \draw[edge] (R12) -- (R13) node[midway,right] {3};
  \draw[edge] (R12) -- (R14) node[midway,above,pos=0.2] {5};
  \draw[edge] (R13) -- (R14) node[midway,above] {4};
        \end{tikzpicture}
      }
    &
      \makebox{
        \begin{tikzpicture}[baseline=(current bounding box.center),
    vertex/.style={circle,draw,fill=blue!20,minimum size=15pt,inner sep=0pt},
    edge/.style={thick},
    scale=1.5]
  \node[vertex] (R21) at (0,0) {${f_1}$};
  \node[vertex] (R22) at (1.4,0) {${f_3}$};    
  \node[vertex] (R23) at (1.4,1) {${f_4}$}; 
  \node[vertex] (R24) at (0,1) {${f_2}$}; 
  \node[vertex] (R25) at (2.8,1) {${f_6}$};
  \node[vertex] (R26) at (2.8,0) {${f_5}$};    
  \draw[edge] (R21) -- (R22) node[midway,above] {4};
  \draw[edge] (R21) -- (R23) node[midway,above,pos=0.2] {5};
  \draw[edge] (R21) -- (R24) node[midway,left] {3};  
  \draw[edge] (R22) -- (R23) node[midway,right,pos=0.7] {3};
  \draw[edge] (R22) -- (R24) node[midway,above,pos=0.2] {5};
  \draw[edge] (R23) -- (R24) node[midway,below] {4};
  \draw[edge] (R22) -- (R26) node[midway,above] {4};
  \draw[edge] (R23) -- (R25) node[midway,below] {4};
  \draw[edge] (R22) -- (R25) node[midway,above,pos=0.1] {5};
  \draw[edge] (R23) -- (R26) node[midway,above,pos=0.7] {5};
  \draw[edge] (R25) -- (R26) node[midway,right] {3};
  \draw[edge] (R21) to[out=330,in=210] node[midway,above,pos=0.3] {8} (R26);
  \draw[edge] (R24) to[out=30,in=150] node[midway,below,pos=0.7] {8} (R25);
        \end{tikzpicture}
        }
    \end{tabular}
\end{center}
\caption{The gadget graphs $R_1$ (left) and $R_2$ (right) show that $\mathcal{G}_{1}$ is a YES instance of the EDGP${}_{1}$ iff $\mathcal{G}_2$ is a YES instance of the EDGP${}_2$.}
\label{fig:saxeR}
\end{figure}

Saxe does not warn the reader that the graph $\mathcal{G}'$, after the replacement with $R_1,R_2$, has edge weights scaled by $4$: edges of $\mathcal{G}_1$ with weight $1$ (resp.~$2$) become edges with weight $4$ (resp.~$8$) in $\mathcal{G}'$. Neither does he explicitly provide the gadgets $T_5,T_8$ to reduce edges with weights $5,8$ to edges with weights $1,2$: it suffices to grow $T_3,T_4$ to $T_5,T_8$ by simply attaching new vertices linked by two previous ones: we need one new vertex to construct $T_5$ from $T_4$, and three more to construct $T_8$ from $T_5$ (see Fig.~\ref{fig:saxeT58}). The graph after these $T_h$ replacements identifies the vertices $u,v$ adjecent to an edge $\{u,v\}$ with weight $h$ with $e_1,e_{h+1}$ in $T_h$ (for $h\in\{3,4,5,8\}$).
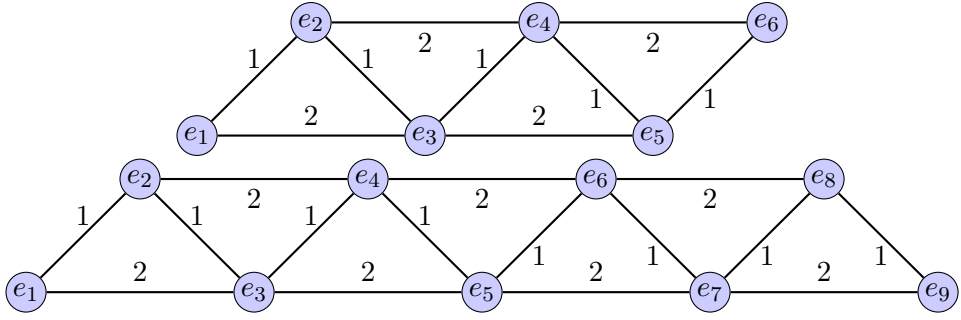
\begin{figure}[!ht]
  \begin{center}
\begin{tikzpicture}[
    vertex/.style={circle,draw,fill=blue!20,minimum size=15pt,inner sep=0pt},
    edge/.style={thick},
    scale=1.5]
  \node[vertex] (T31) at (0,0) {${e_1}$};
  \node[vertex] (T32) at (2,0) {${e_3}$};    
  \node[vertex] (T33) at (1,1) {${e_2}$}; 
  \node[vertex] (T34) at (3,1) {${e_4}$};
  \node[vertex] (T35) at (4,0) {${e_5}$};  
  \node[vertex] (T36) at (5,1) {${e_6}$};  
  \draw[edge] (T31) -- (T32) node[midway,above] {2};
  \draw[edge] (T31) -- (T33) node[midway,above] {1};
  \draw[edge] (T32) -- (T33) node[midway,above] {1};  
  \draw[edge] (T32) -- (T34) node[midway,above] {1};
  \draw[edge] (T32) -- (T35) node[midway,above] {2};
  \draw[edge] (T33) -- (T34) node[midway,below] {2};
  \draw[edge] (T34) -- (T35) node[midway,below] {1};
  \draw[edge] (T34) -- (T36) node[midway,below] {2};
  \draw[edge] (T35) -- (T36) node[midway,below] {1};
\end{tikzpicture} \\
\begin{tikzpicture}[
    vertex/.style={circle,draw,fill=blue!20,minimum size=15pt,inner sep=0pt},
    edge/.style={thick},
    scale=1.5]
  \node[vertex] (T31) at (0,0) {${e_1}$};
  \node[vertex] (T32) at (2,0) {${e_3}$};    
  \node[vertex] (T33) at (1,1) {${e_2}$}; 
  \node[vertex] (T34) at (3,1) {${e_4}$};
  \node[vertex] (T35) at (4,0) {${e_5}$};  
  \node[vertex] (T36) at (5,1) {${e_6}$};
  \node[vertex] (T37) at (6,0) {${e_7}$};
  \node[vertex] (T38) at (7,1) {${e_8}$};
  \node[vertex] (T39) at (8,0) {${e_9}$};
  \draw[edge] (T31) -- (T32) node[midway,above] {2};
  \draw[edge] (T31) -- (T33) node[midway,above] {1};
  \draw[edge] (T32) -- (T33) node[midway,above] {1};  
  \draw[edge] (T32) -- (T34) node[midway,above] {1};
  \draw[edge] (T32) -- (T35) node[midway,above] {2};
  \draw[edge] (T33) -- (T34) node[midway,below] {2};
  \draw[edge] (T34) -- (T35) node[midway,above] {1};
  \draw[edge] (T34) -- (T36) node[midway,below] {2};
  \draw[edge] (T35) -- (T36) node[midway,below] {1};
  \draw[edge] (T35) -- (T37) node[midway,above] {2};
  \draw[edge] (T36) -- (T37) node[midway,below] {1};
  \draw[edge] (T36) -- (T38) node[midway,below] {2};
  \draw[edge] (T37) -- (T38) node[midway,below] {1};
  \draw[edge] (T37) -- (T39) node[midway,above] {2};
  \draw[edge] (T38) -- (T39) node[midway,below] {1};  
\end{tikzpicture} 
\end{center}
\caption{The gadget graphs $T_5$ (above) and $T_8$ (below) used in proving the \textbf{NP}-hardness of EDGP${}_K$.}
\label{fig:saxeT58}
\end{figure}

\subsection{Saxe's realized gadgets $R_1,R_2$ only span $\mathbb{R}^2$}
We think there is a more serious issue with the $R_1,R_2$ gadgets, however, in that they are all realizable in $\mathbb{R}^2$. The induction starts off with a valid reduction from EDGP${}_1$ to EDGP${}_2$, but then, at the next step, $\mathcal{G}_3$ must also be realizable in $\mathbb{R}^2$ as long as $\mathcal{G}_2$ is realizable in $\mathbb{R}^2$ (because all of the added gadgets $R_1,R_2$ are realizable in $\mathbb{R}^2$), and so, inductively, $\mathcal{G}_K$ is realizable in $\mathbb{R}^2$ too. This does not make Saxe's result false, strictly speaking, since if a graph is realizable in $\mathbb{R}^2$ then it is also realizable  in $\mathbb{R}^K$ for all $K>1$: either trivially, by considering $\mathbb{R}^2$ a subspace of $\mathbb{R}^K$, or non-trivially, by \textit{choosing} to realize $R_1,R_2$ at the $K$-th inductive step in a $2D$ subspace orthogonal to $\mathbb{R}^{K-1}$. But we think Saxe's result is disappointing for $K>2$ in the sense that each $\mathcal{G}_K$ for $K>2$ is a YES instance of EDGP${}_\kappa$ for $\kappa\in\{2,\ldots,K\}$. We feel that the result would be much more convincing if each $\mathcal{G}_K$ \textit{needed} all of the $K$ dimensions in $\mathbb{R}^K$ in order to be realized. In other words, if $\mathcal{G}_K$ were a YES instance of EDGP${}_K$ but a NO instance of EDGP${}_{\kappa}$ for all $\kappa<K$.

\subsection{Simpler gadgets that span $\mathbb{R}^K$ and their critique}
Given any $K>1$, we therefore propose a direct (i.e.~non-inductive) construction of $\mathcal{G}_K$ based on replacing a single edge $\{u,v\}$ of $\mathcal{G}_1$, say with weight $1$, with a clique gadget $C^K_1$ over $K+1$ vertices with uniform edge weights having value $1$ (see Fig.~\ref{fig:klC} for a definition of $C^2_w$ and $C^3_w$ for any $w>0$), in such a way that $u\equiv k_1$ and $v\equiv k_2$. Since, for any $w>0$, the gadget $C^K_w$ can only be realized in $\mathbb{R}^K$ but not in $\mathbb{R}^{\kappa}$ for any $\kappa<K$, it follows that $\mathcal{G}_1$ is a YES instance of EDGP${}_1$ iff $\mathcal{G}_K$ is a YES instance of EDGP${}_K$. 
\begin{figure}[!ht]
  \begin{center}
    \begin{tikzpicture}[
        vertex/.style={circle,draw,fill=blue!20,minimum size=15pt,inner sep=0pt},
        edge/.style={thick},
        scale=1.5]
      \node[vertex] (Cw1) at (0,0) {${k_1}$};
      \node[vertex] (Cw2) at (2,0) {${k_2}$};    
      \node[vertex] (Cw3) at (1,1.732) {${k_3}$}; 
      \draw[edge] (Cw1) -- (Cw2) node[midway,above] {$w$};
      \draw[edge] (Cw2) -- (Cw3) node[midway,left] {$w$};
      \draw[edge] (Cw3) -- (Cw1) node[midway,right] {$w$};  
    \end{tikzpicture}\qquad 
    \begin{tikzpicture}[
        vertex/.style={circle,draw,fill=blue!20,minimum size=15pt,inner sep=0pt},
        edge/.style={thick},
        scale=1.5]
      \node[vertex] (Cw1) at (0,0) {$k_1$};
      \node[vertex] (Cw2) at (1.732,0) {$k_2$};    
      \node[vertex] (Cw3) at (1.732,1.732) {$k_3$}; 
      \node[vertex] (Cw4) at (0,1.732) {$k_4$}; 
      \draw[edge] (Cw1) -- (Cw2) node[midway,above] {$w$};
      \draw[edge] (Cw1) -- (Cw3) node[midway,above,pos=0.2] {$w$};
      \draw[edge] (Cw1) -- (Cw4) node[midway,left] {$w$};  
      \draw[edge] (Cw2) -- (Cw3) node[midway,right] {$w$};
      \draw[edge] (Cw2) -- (Cw4) node[midway,above,pos=0.2] {$w$};
      \draw[edge] (Cw3) -- (Cw4) node[midway,above] {$w$};
    \end{tikzpicture}
  \end{center}
  \caption{Clique gadgets $C^K_w$ for $K\in\{2,3\}$.}
\label{fig:klC}
\end{figure}
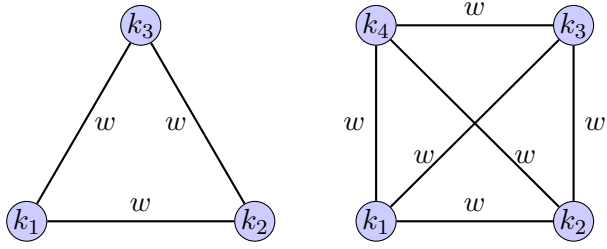
We remark however that this definition of $\mathcal{G}_K$ yields a flexible (rather than rigid) framework for any $K>2$, since different gadgets $C^K_w$ can rotate around the line spanned by the original 1D realization of $\mathcal{G}_1$ with different angles (see Fig.~\ref{fig:flex}). Since realization uniqueness (related to rigidity) will be used in Sect.~\ref{s:ambiguous}, this simple construction is only valid for the task at hand. 
\begin{figure}[!ht]
  \begin{center}
    \includegraphics[width=5cm]{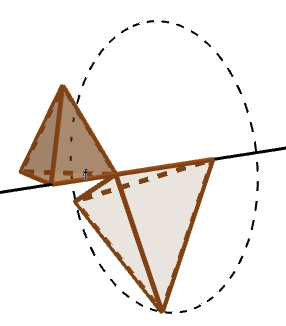}
  \end{center}
  \caption{Two clique gadgets $C^3_w$, for different values of $w$: each clique gadget can independently rotate around the reference axis containing the 1D realization of $\mathcal{G}_1$, yielding a flexible framework.}
\label{fig:flex}
\end{figure}

\subsection{Gadgets that fill $\mathbb{R}^K$ and yield a rigid graph}
In order to obtain a rigid framework, we propose an extension of the graph gadgets $R_1,R_2$ to any $K$ rather than just $K=2$. We consider $R_1$ as two $2$-cliques with uniform weight $3$ (the first consisting of vertices $f_1,f_2$, the second of $f_3,f_4$) joined by edges $\{f_1,f_3\}$ and $\{f_2,f_4\}$ with weight $4$, and all of the other edges having weight $5$. The generalization of this construction in $K$ dimension consists of two $K$-cliques with all edges having weight $3$, joined by all edges $\{i,K+i\}$ (for $i\le K$) with weight $4$, and all other edges with weight $5$. We call this gadget graph $\bar{R}^K_1$. We define $\bar{R}^K_2$ from two copies of $\bar{R}^K_1$ analogously to $R_2$ with respect to $R_1$: two copies of $\bar{R}^K_1$ are joined by merging the second $K$-clique of the first copy to the first $K$-clique of the second copy, which gives a graph on $3K$ vertices labeled $\{1,\ldots,K\}$ (first $K$-clique), $\{K+1,\ldots,2K\}$ (second $K$-clique), $\{2K+1,\ldots 3K\}$ (third $K$-clique); finally, we add edges $\{i,2K+i\}$ for all $i\le K$ with weight $8$.

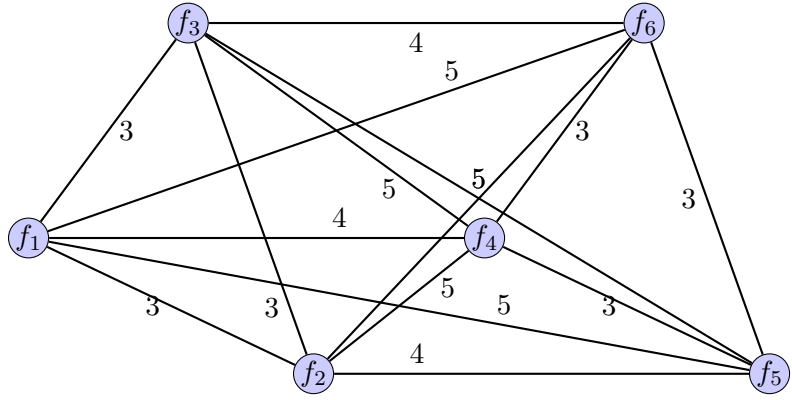
\begin{figure}[!ht]
  \begin{center}
    \begin{tikzpicture}[
        vertex/.style={circle,draw,fill=blue!20,minimum size=15pt,inner sep=0pt},
        edge/.style={thick},
        scale=1.5]
      \node[vertex] (RK1) at (0,0) {${f_1}$};
      \node[vertex] (RK2) at (4,0) {${f_4}$};    
      \node[vertex] (RK3) at (5.4,1.9) {${f_6}$}; 
      \node[vertex] (RK4) at (1.4,1.9) {${f_3}$};
      \node[vertex] (RK5) at (2.5,-1.2) {${f_2}$};
      \node[vertex] (RK6) at (6.5,-1.2) {${f_5}$};
      \draw[edge] (RK1) -- (RK2) node[midway,above,pos=0.7] {4};
      \draw[edge] (RK2) -- (RK3) node[midway,right] {3};
      \draw[edge] (RK3) -- (RK4) node[midway,below] {4};
      \draw[edge] (RK4) -- (RK1) node[midway,right] {3};
      \draw[edge] (RK1) -- (RK3) node[midway,above,pos=0.7] {5};
      \draw[edge] (RK1) -- (RK5) node[midway,left] {3};
      \draw[edge] (RK2) -- (RK6) node[midway,left] {3};
      \draw[edge] (RK5) -- (RK6) node[midway,above,pos=0.2] {4};
      \draw[edge] (RK1) -- (RK6) node[midway,above,pos=0.65] {5};
      \draw[edge] (RK4) -- (RK5) node[midway,left,pos=0.85] {3};
      \draw[edge] (RK3) -- (RK6) node[midway,left] {3};
      \draw[edge] (RK5) -- (RK3) node[midway,above] {5};
      \draw[edge] (RK4) -- (RK2) node[midway,below,pos=0.7] {5};
      \draw[edge] (RK5) -- (RK2) node[midway,below,pos=0.85] {5};
      \draw[edge] (RK4) -- (RK6) node[midway,above] {5};
    \end{tikzpicture}    
  \end{center}
  \caption{Gadget $\bar{R}^K_1$ for $K=3$.}
\label{fig:RK1}
\end{figure}

The framework $(\bar{R}^K_1,x)$ where $x$ is any distance-respecting realization is globally rigid in $\mathbb{R}^K$ because it is a $2K$-clique that realizes in $\mathbb{R}^K$, and it cannot realize in any fewer dimensions because of the two $K$-cliques with all edges having weight $3$. We remark that: (i) $\bar{R}^K_1$ is integer-weighted, like $R_1$, (ii) in the realization of $\bar{R}^K_1$ the segments with length $4$ are all parallel, and so they can be arranged to span the whole real line, (iii) $\bar{R}^K_1$ can be realized in $\mathbb{R}^K$ but not in $\mathbb{R}^{\kappa}$ for any $\kappa<K$ (the two $K$-cliques can each be realized in $\mathbb{R}^{K-1}$, and the other edges can all be realized $K$ dimensions). We note that $\bar{R}^K_1$ shares properties (i)-(ii) with gadget $R_1$ (Fig.~\ref{fig:saxeR}, left), which allows us to use $\bar{R}^K_1$ as a replacement for edges having weight $1$ in $\mathcal{G}_1$, in the same way as $R_1$ is used. The gadget $\bar{R}^K_2$ is analogous to $R_2$.

The modified gadgets $\bar{R}^K_1,\bar{R}^K_2$ remain polynomially sized (each has $O(K)$ vertices and $O(K^2)$ edges), are applied to $\mathcal{G}_1$ in a similar way as Saxe's original gadgets $R_1,R_2$, and ensure that the the replacement of edges of weight $1$ (resp.~$2$) in $\mathcal{G}_1$ with $\bar{R}^K_1$ (resp.~$\bar{R}^K_2$) yields a graph $\bar{\mathcal{G}}'$ that can be realized in $\mathbb{R}^K$ but not in any $\mathbb{R}^\kappa$ for $\kappa<K$. Like $R_1,R_2$, each of the modified gadgets $\bar{R}^K_1,\bar{R}^K_2$ can be placed next to each other so that the rightmost clique of the left gadget coincides with the leftmost clique of the right gadget. This ensures that the resulting graph $\bar{\mathcal{G}}'$ can still be constructed similarly to $R_1,R_2$ while being rigid in $\mathbb{R}^K$ (differently from the clique gadgets $C^K_w$).

The graph $\bar{\mathcal{G}}_K$ is derived from $\bar{\mathcal{G}}'$ by replacing every edge with weight $w$ with the corresponding $T_w$ gadget, so that $\bar{\mathcal{G}}_K$ only has edges with weights in $\{1,2\}$. We conclude that $\mathcal{G}_1$ is a YES instance of EDGP${}_1$ iff $\bar{\mathcal{G}}_K$ is a YES instance of EDGP${}_K$.

\subsection{Critique of the new gadgets}
There is a notable difference between our construction and Saxe's, aside from the fact that Saxe's only really allows a true dimensional increase up to $K=2$: Saxe's construction can be certified by integer realizations, while ours cannot: realizing each of the two $K$-cliques requires at least one point to be irrational (albeit algebraic). Replacing these cliques by hypercubes would fix the issue but make the increase in graph vertex size exponential (the $K$-dimensional hypercube has $2^K$ vertices). This will have an impact in Sect.~\ref{s:realK}.

\section{EDGP with real-valued edge weights}
\label{s:realEDGP}
In \cite[\S 5]{saxe79}, Saxe discusses the case of the EDGP where the weight function $d:E\to\mathbb{R}$ is not restricted to integer values (as in all of the above cases), but can actually take real values. To discuss the \textbf{NP}-hardness of this case, he introduces the following definitions: (i) for an $\epsilon>0$, an $\epsilon$-approximate realization $x$ is such that
\begin{equation}
  \forall \{u,v\}\in E \quad (1-\alpha)d_{uv}\le \|x_u - x_v\|_2^2 \le (1+\alpha)d_{uv}\,; \label{epsdgp}
\end{equation}
(ii) for $0<\epsilon<\delta$ the $(\epsilon,\delta)$-approximate EDGP consists in asserting whether the real-weighted graph $G=(V,E,d)$ has a $\delta$-approximate realization (YES instance) or has no $\epsilon$-approximate realization in $\mathbb{R}^K$ (NO instance); (iii) a graph $G$ is $\epsilon$-approximately realizable if $G$ has a realization satisfying Eq.~\eqref{epsdgp}. According to Saxe, if the least value $\epsilon'$ for which $G$ is $\epsilon'$-approximately realizable is such that $\epsilon'\in (\epsilon,\delta)$, then $G$ can be arbitrarily declared either a YES or a NO instance. Next, he states a theorem that the $(1/18,1/9)$-approximate EDGP${}_1$ where $\ran{d}$ is integer is \textbf{NP}-complete. The ``sketch of proof'' is as follows:
\begin{quote}
  {\small
    ``We note that the embeddability properties of the graphs used in the proof of \cite[Thm.~4.1]{saxe79} depend only on cycles of length no greater than 16 having edges whose length are multiples of 1. It follows from this that for such graphs $\epsilon$-approximate $1$-embeddability is equivalent to ordinary $1$-embeddability for any $\epsilon<1/8$, and the desired result is at hand.'' ($\ast$)
  }
\end{quote}
In \cite[\S 5]{saxe80}, Saxe makes three changes to the above description: (1) the range of instances of $(\epsilon,\delta)$-approximate EDGP${}_1$ for which the YES/NO status is arbitrarily decided goes from the closed interval $[\epsilon,\delta]$ to the half-open interval {$(\epsilon,\delta]$}; (2) the theorem asserts the \textbf{NP}-completeness of the $(\epsilon,\delta)$-approximate EDGP${}_1$ with integer weights for all $0\le\epsilon\le\delta<1/8$ (as opposed to $\epsilon<\delta$); (3) the ``sketch of proof'', \textit{verbatim}, becomes a ``proof''.\footnote{Saxe makes two more additions in \cite[\S 5]{saxe80}: (i) a theorem stating that there is $\epsilon>0$ such that the $(0,\epsilon)$-approximate EDGP${}_K$ where $\ran{d}$ is integer is \textbf{NP}-hard (the proof simply says ``the argument is outlined in the above text. Details are left to the reader'', where the text points to gadgets $R_1,R_2$ in Fig.~\ref{fig:saxeR}), and (ii) a comment stating that an intervention of the audience during the Allerton conference (whose proceedings include \cite{saxe79}), added that the $(\epsilon,\delta)$-approximate EDGP${}_K$ is actually in \textbf{NP}. We will make sense of (i) but leave (ii) out, since it is not due to Saxe.}

  Our main criticism is that Saxe's purpose, that of studying graphs with real-valued weights, was turned into another goal altogether, namely that of studying graphs with integer weights realized by algorithms that provide imprecise realizations. Our second criticism concerns the fact that the proof of \cite[Thm.~5.1]{saxe79} is excessively hermetic. The third criticism concerns the extension of the \textbf{NP}-hardness of the $(\epsilon,\delta)$-approximate EDGP${}_1$ with integer weights to higher dimensions $K>1$.

  \subsection{Graphs with real-valued integer realizations}
  \label{s:realint}
  There is not much more to say about our main criticism: all \textbf{NP}-hardness proofs in \cite[\S 5]{saxe79} and \cite[\S 5]{saxe80} clearly specify that $\ran{d}$, i.e.~the edge weights, are integer-valued, and therefore the section title about real-valued edge weights is misleading. Indeed, edge weight integrality is one of the keys to understanding the hermetic proof of \cite[Thm.~5.1]{saxe79}, as we shall see. Therefore Saxe's focus in this section is in the $\epsilon$-approximate realization of an integer-weighted graph in 1 dimension. While this is certainly a topic worth studying, we fail to see why one would want to use imprecise floating-point arithmetic in order to find a solution that is known \textit{a priori} to be integer. 

  \subsection{\textbf{NP}-hardness of approximate realizability}
  \label{s:realapprox}
  We call the sum of edge weights over a cycle the \textit{cycle length}. The proof of \textbf{NP}-hardness of the $(\epsilon,\delta)$-approximate EDGP${}_1$, as stated by Saxe, is as follows: the graph $\mathcal{G}$ formed by the gadget graphs in Fig.~\ref{fig:saxe12}-\ref{fig:saxe3} has cycles having length at most $16$; in these cycles there are also edges with weight $1$; therefore any graph with these properties is $\epsilon$-approximately realizable iff it is realizable in $K=1$ dimensions. Here is a list of our doubts about this theorem of Saxe.
  \begin{enumerate}
  \item The definition of the $(\epsilon,\delta)$-approximate EDGP${}_1$ problem is unusual, in that some instances need to be assigned a YES/NO status arbitrarily. \label{doubt1} 
  \item In order to prove \textbf{NP}-hardness of the $(\epsilon,\delta)$-approximate EDGP${}_1$, Saxe appears to suggest that he is reducing EDGP${}_1$ to its approximate version. If this were so, one should take \textit{any} instance $G=(V,E,d)$ of the EDGP${}_1$ and construct an instance of the approximate version that is YES if and only if $G$ is YES in in EDGP${}_1$. But Saxe constrains instances to have cycles with cumulative weight $16$ and edges with weight $1$. \label{doubt2}
  \item Saxe further implies that it is sufficient to reduce EDGP${}_1$ graphs consisting of a single cycle of length $16$, but the preceding doubt suggests that the graph that needs to be reduced to the target problem might consist of many cycles: why does the proof not need to show that \textit{all} of the cycles in ths graph can be realized consistently? \label{doubt3}
  \item Why does the cycle length $16$ induce a $\delta=1/8$? More generally, what is the relationship between the cycle length and $\delta$? \label{doubt4}
  \item Since the $(\epsilon,\delta)$-approximate EDGP${}_1$ has a condition for NO that is not the negation of the condition for YES, the NO instances of SAT that map to an unrealizable $\mathcal{G}$ should be also mentioned in the proof, but this is not so. \label{doubt5}
  \end{enumerate}

  \subsubsection{An \textbf{NP}-complete subclass of EDGP${}_1$.}
  We first deal with our doubt \ref{doubt2}: Saxe's assumptions can hold only if we actually reduce \textsc{3sat} to the $(\epsilon,\delta)$-approximate EDGP${}_1$, via the EDGP${}_1$ instance graph $\mathcal{G}$ constructed using the gadgets in Fig.~\ref{fig:saxe12}-\ref{fig:saxe3}. So this is not a reduction from EDGP${}_1$ to its approximate counterpart, but a reduction from the \textbf{NP}-complete subclass of EDGP${}_1$ instance graphs $\mathcal{G}$. This is of course sufficient to prove the \textbf{NP}-hardness of the target problem, but a reduction from \textit{any} instance of EDGP${}_1$ would have been more convincing.

  \subsubsection{Cycles with length $16$.}
  Next, we discuss doubt \ref{doubt4} (we hope that this discussion will also clarify doubts \ref{doubt1} and \ref{doubt5}). Recall the first hardness result in \cite{saxe79}, i.e.~the weak \textbf{NP}-hardness of the EDGP${}_1$ by reduction from \textsc{partition}. YES instances $(a_1,\ldots,a_n)$ of \textsc{partition} reduce to single cycle graphs having edges $\{v,v+1\}$ with length $d_{v,v+1}=a_v$ that can be realized so that the cycle can be closed. The cycle length is $L=\sum_{v<n} a_v$. Vertices $v$ can be placed either left or right of the previous vertex $v-1$. For YES instances of \textsc{partition}, the left-hand side sum is equal to the right-hand-side sum:
  \begin{equation}
    \sum_{v\in I} a_v=\sum_{v\not\in I} a_v. \label{eq:cyclesum}
  \end{equation}
  The right-hand side is the sum of the lengths of edges realized with an orientation from $v$ to $v+1$ (placement of $v+1$ on the right of $v$). The left-hand side is the same with an orintation from $v+1$ to $v$ (placement of $v+1$ on the left of $v$). Therefore each side of the equation must sum to half of the cycle length, i.e.~$L/2$. Since both sums are integer, it follows that $L$ must be even. Back to the graph $\mathcal{G}$, if the longest cycles have length $L=16$, each sum must have value $L/2=8$.

  Assume that the EDGP${}_1$ instance graph $\mathcal{G}$ is YES, $\delta=1/8$, and there exists an approximate realization $y$ such that $y_{v+1}-y_{v}=(1+1/8)d_{v,v+1}$ for all $v\in I$. Then we have
  \begin{equation}
    \sum_{v\in I} (1+1/8)d_{v,v+1}=(1+1/8)\sum_{v\in I}d_{v,v+1}=(1+1/8)\times 8=9,  \label{eq:halflen}
  \end{equation}
  which is an integer value different from $L/2=8$. As a possible consequence of this integer error, consider the NO \textsc{partition} instance $(8,9)$: the corresponding $(\epsilon,1/8)$-approximate EDGP${}_1$ instance is YES, as witnessed by the realization that places $v_1=0$, $v_2=(1+1/8)\times 8=9$, and $v_3=v_1=0$, which explains why reductions from $\mathcal{G}$ to $(\epsilon,\delta)$-approximate EDGP${}_1$ with $\delta\ge 1/8$ may be invalid. 
  By contrast, if $\delta<1/8$, the corresponding value of Eq.~\eqref{eq:halflen} would be fractional, and belong to the open interval $(7,9)$. In such a situation it is possible to derive an exact realization from an approximate one, which provides the key observation for a valid reduction from EDGP${}_1$ to its approximate version.

  We remark that the fractional edge error $1/8$ was derived from $2/16$, where $16$ is the cycle length: this generalizes to $\delta\ge 2/L$ for cycles of length $L$ with $n$ edges with an approximate realization $y$. For the $v$-th edge $\{v,v+1\}$ of such a cycle let $\alpha_v=1-\frac{\|y_{v+1}-y_v\|_2}{d_{v,v+1}}$ and let $\Phi=\sum_{v\in I}\alpha_v$. If $0\le\Phi<1$ then a judicious rounding of the errors can lead to a correct realization (see below). On the other hand, there may be cycles reduced from NO instances of \textsc{partition} that can be approximately realized while yielding $\Phi\ge 1$, which invalidates the reduction. In general, for cycles of even length $L$, the extent of realization imprecision may only vary in the open interval $(1-2/L,1+2/L)$, beyond which the reduction may be invalid.

  We note that, in the above reasoning, we only mentioned \textsc{partition} to explain why the cycle length $L=16$ induces a $\delta=2/L=1/8$. The reduction of interest, however, is not from \textsc{partition} but from \textsc{3sat}, via the EDGP${}_1$ instance graph $\mathcal{G}$ (see Sect.~\ref{proof51}). The point is that, if $\mathcal{G}$ is a YES instance of the EDGP${}_1$, each cycle in its graph will have a realization $x$ satisfying Eq.~\eqref{eq:halflen}: in particular, this will hold also for the longest cycles that have length $L=16$. These yield the smallest upper bound on $\delta$, i.e.~$\delta<2/L=1/8$. Since $x$ is a precise realization, it is also valid for the $(\epsilon,\delta)$-approximate EDGP${}_1$ for any $0<\epsilon<\delta$: this establishes that YES instances of the source problem correspond to YES instances of the target problem.

  If $\mathcal{G}$ is a NO instance, however, there will be no precise realization of at least some of its cycles. We suppose, to aim at a contradiction, that there are $0<\epsilon<\delta$ such that the corresponding $(\epsilon,\delta)$-instance is YES. Then there must be a $\delta$-approximate realization $y$ that realizes each cycle of $\mathcal{G}$ so that edges with weights $d_{v,v+1}$ are realized as segments of length $|y_{v+1}-y_v|$ in the interval $[(1-\delta)d_{v,v+1},(1+\delta)d_{v,v+1}]$. Now we define a realization $x$ from $y$ so that it corrects the fractional error introduced by $y$, and therefore realizes the cycle exactly, as follows:
  \begin{enumerate}
    \setlength{\parskip}{-0.3em}
  \item if $y_{v+1}>y_v$ and $y_{v+1}-y_v>d_{v,v+1}$ let $x_{v+1}=\lfloor y_{v+1}\rfloor$\label{ae1}
  \item if $y_{v+1}>y_v$ and $y_{v+1}-y_v<d_{v,v+1}$ let $x_{v+1}=\lceil y_{v+1}\rceil$\label{ae2}
  \item if $y_{v+1}>y_v$ and $y_{v+1}-y_v=d_{v,v+1}$ let $x_{v+1}=y_{v+1}$\label{ae3}
  \item if $y_{v+1}<y_v$ and $y_v-y_{v+1}>d_{v,v+1}$ let $x_{v+1}=\lceil y_{v+1}\rceil$\label{ae4}
  \item if $y_{v+1}<y_v$ and $y_v-y_{v+1}<d_{v,v+1}$ let $x_{v+1}=\lfloor y_{v+1}\rfloor$\label{ae5}
  \item if $y_{v+1}<y_v$ and $y_v-y_{v+1}=d_{v,v+1}$ let $x_{v+1}=y_{v+1}$.\label{ae6}
  \end{enumerate}
  The above algorithm is feasible because the integer distance values $d_{uv}$ of the graph $\mathcal{G}$ are part of the approximate EDGP${}_1$ input. We remark that such an $x$ realizes any cycle of length $L$ exactly as long as the following condition holds:
\begin{center}
  $x_{v+1}\not=y_{v+1}$ in cases \ref{ae1}-\ref{ae2} and \ref{ae4}-\ref{ae5}.
\end{center}
If this condition holds, then each $x_v$ loses the fractional part of $y_v$ and is reset to the integer value such that $|x_{v+1}-x_v|=d_{v,v+1}$, which yields a valid realization $x$ of the NO instance of the EDGP${}_1$: a contradiction. Therefore there are cycles in $\mathcal{G}$ such that no $\delta$-approximate realization of the cycle exists. This, in turn, implies that no realizations for $\mathcal{G}$ can exist, which establishes that NO instances of the source problem correspond to NO instances of the target problem. This arguement proves the \textbf{NP}-hardness of $(\epsilon,\delta)$-approximate EDGP${}_1$.

Let us also pursue the case in which the above condition does not hold: then $y_{v+1}$ is an integer in some cases \ref{ae1}-\ref{ae2} or \ref{ae4}-\ref{ae5}. This implies that for such vertices we have $x_{v+1}=y_{v+1}$, and we can no longer deduce that $|x_{v+1}-x_v|=d_{v,v+1}$. In fact, since both $x_{v+1},x_v$ are integers, we have  $|\,|x_{v+1}-x_v|-d_{v,v+1}|>1$, resulting in $\delta>1$. Moreover, there can be no realization with error $\le\epsilon<2/L$ since otherwise the above rounding algorithm would obtain a valid realization. This allows us to deduce directly (instead of by contradiction) that the $(\epsilon,\delta)$-approximate EDGP${}_1$ instance is a NO instance. We think that this observation about $\epsilon$ should provide an answer to doubts \ref{doubt1} and \ref{doubt5}: the definition of $(\epsilon,\delta)$-approximate EDGP${}_1$ certainly looks ambiguous, but it is precise enough to support the reduction and its analysis.

  \subsubsection{From the cycles to the graph} \label{proof51}
  Lastly, we tackle doubt \ref{doubt3}: consider a smallest cycle cover $\mathcal{C}$ of $\mathcal{G}$ constrained to contain all cycles of length $16$. If the EDGP${}_1$ instance is YES then $\mathcal{G}$ has a realization $x$. Each cycle in $\mathcal{C}$ is also realized by the relevant components of $x$, and so, by the above discussion, $x$ also provides a YES certificate for $\mathcal{G}$ as an instance of $(\epsilon,\delta)$-approximate EDGP${}_1$ for any $0<\epsilon<\delta$. If $\mathcal{G}$ is a NO instance of EDGP${}_1$, then, by the gadget in Fig.~\ref{fig:saxe3} and the above discussion, there must be at least one cycle of length $16$ that is not realized correctly, making $\mathcal{G}$ a NO instance of $(\epsilon,\delta)$-approximate EDGP${}_1$. We hope that this provides a more comprehensible proof of \cite[Thm.~5.1]{saxe79}. 
  
\subsection{Extension to higher dimensions}  
\label{s:realK}
The extension of Saxe's proof of the approximate EDGP${}_1$ to the approximate EDGP${}_K$ for $K>1$ is similar to the extension of the proof of \textbf{NP}-hardness of EDGP${}_1$ extended to EDGP${}_K$ with $K>1$: it is based on gadgets $R_1,R_2$ and $T_k$ for $k\in \{3,4,5,8\}$. It has the same shortcoming as the extension to EDGP${}_K$ in the sense that for $K>2$ the realizations are still in an embedding of $\mathbb{R}^2$ rather than needing all $K$ dimensions. Our modification of the gadgets $R_,R_2$ into $\bar{R}^K_1,\bar{R}^K_2$ (see Fig.~\ref{fig:RK1}) cannot be adapted to the proof of Sect.~\ref{s:realapprox}, however, because at least two vertices in each $\bar{R}^K_1$ and at least three vertices in each $\bar{R}^K_2$ need to be placed at irrational positions whenever $K>2$ (the irrational vertices belong to the initial $K$-cliques, as mentioned at the end of Sect.~\ref{s:edgpK}).

\section{The ambiguous EDGP}
\label{s:ambiguous}
This section refers to \cite[\S 6]{saxe79} and \cite[\S 6]{saxe80}. We think that the definitions are understandable in both papers. While \cite{saxe79} only contains a proof sketch, \cite{saxe80} contains a complete proof. Although many details are skipped, we think that the missing details are not excessively hard to infer. We therefore simply give an overview of Saxe's material. 

Let $P$ be a problem in \textbf{NP}. By the definition of \textbf{NP} there exists a certificate $c$ such that the verifier $V_P(p,c)$ yields $\mathsf{TRUE}$ if and only if $p$ is a YES instance. The \textit{ambiguous} variant of $P$ ascertains whether a YES instance of $p$ has multiple certificates: given a YES instance $p$ of $P$ and a certificate $c$ of $p$ such that $V_P(p,c)=\mathsf{TRUE}$, is there another certificate $c'\not=c$? 

In the context of the EDGP, we assume that certificates are realizations of the instance graphs. We consider two realizations $x,y$ ``equal'' if they are congruent, i.e.~there is a congruence mapping one to the other. We introduce the notation $x\approx y$ to mean that $x,y$ are congruent. For $K=1$, congruences only involve rational arithmetic, and verifying whether $x\approx y$ can be done in polynomial time. For $K>1$ this is not the case, because rotation matrices may involve irrational (and non-algebraic) numbers. In such cases the verification of $x\approx y$ can be done in polynomial time on a real RAM.

We can finally define the A-EDGP: given $K>0$, a weighted graph $G=(V,E,d)$, and an $n\times K$ matrix $x$ that is a realization of $G$, is there another realization $y$ of $G$ such that $x\not\approx y$? Saxe's reduction follows multiple steps: he first reduces \textsc{3sat} to the ambiguous \textsc{4sat} (\textsc{a-4sat}), which is then reduced to the ambiguous \textsc{3sat} (\textsc{a-3sat}), which is further reduced to the A-EDGP${}_1$ by Sect.~\ref{s:npcomplete}, and then, by Sect.~\ref{s:edgpK}, to the A-EDGP${}_K$ for $K>1$.

\subsection{From $3$-satisfiability to ambiguous $4$-satisfiability}
The first reduction is as follows: let $\phi$ be a \textsc{3sat} instance in CNF, consisting of clauses $c_1,\ldots,c_m$ each defined on literals based on logical variables $s_1,\ldots,s_n$. We construct an \textsc{a-4sat} instance as follows:
\begin{eqnarray}
  \psi &=& \big(t \wedge \bigwedge_{j\le n} s_j\big) \vee (\bar{t} \wedge \phi) = \big(\bar{t} \vee \bigwedge_{j\le n} s_j\big) \wedge (t\vee \phi) \label{amb1} \\
  &=& \bigwedge_{j\le n} (\bar{t}\vee s_j) \wedge \bigwedge_{i\le m}(t\vee c_i).\label{amb2}
\end{eqnarray}
We claim that $\psi$ in Eq.~\ref{amb1} is always a YES instance of \textsc{sat}. If $\phi$ is satisfiable with certificate $s^\ast$, then we let $\bar{t}=\mathsf{TRUE}$ (i.e.~$t=\mathsf{FALSE}$), which makes the second term in the disjunction satisfiable by the certificate $(\mathsf{FALSE},s^\ast)$: therefore $\psi$ is satisfiable. If $\phi$ is unsatisfiable we let $t=\mathsf{TRUE}$ and assign $s^0_j=\mathsf{TRUE}$ for each $j\le n$. Then $(\mathsf{TRUE},s^0)$ satisfies $\psi$. This settles the claim. Moreover, by Eq.~\eqref{amb2} $\psi$ is an instance of \textsc{4sat}. Now the \textsc{a-4sat} verifier takes $\psi$ and $\phi$ as input, and checks whether the certificate $s^\ast$ for $\phi$ also satisfies $\psi$: if $\phi$ is YES then the certificate $(\mathsf{FALSE},s^\ast)$ also satisfies $\psi$, which means that $\psi$ has another certificate besides $(\mathsf{TRUE},s^0)$, and is therefore a YES instance of \textsc{a-4sat}. If $\phi$ is NO, then any assignment $s'$ of boolean values to the logical variables will not satisfy $\phi$, which means that $\psi$ has the unique solution $(\mathsf{TRUE},s^0)$, implying that $\psi$ is a NO instance of \textsc{a-4sat}. Therefore \textsc{a-4sat} is \textbf{NP}-complete.

\subsection{From ambiguous $4$-satisfiability to ambiguous $3$-satisfiability}
For each $i\le m$, Each clause $c_i$ of an \textsc{4sat} instance $\psi$ consists of four literals by definition: say $c_i=\ell_{i1}\vee\ell_{i2}\vee\ell_{i3}\vee\ell_{i4}$. We introduce a new variable $q_i$ and define the following conjunction of clauses, each consisting of $3$ literals:
\begin{equation}
  c'_i=(\ell_{i1}\vee\ell_{i2}\vee q_i)\land(\ell_{i3}\vee\bar{\ell}_{i4}\vee q_i)\land (\bar{\ell}_{i3}\vee\ell_{i4}\vee q_i)\land (\bar{\ell}_{i3}\vee\bar{\ell}_{i4}\vee q_i)\land (\ell_{i3}\vee\ell_{i4}\vee \bar{q}_i).
\end{equation}
It turns out that, for each assignment of boolean values to $\ell_{ij}$ (for $j\le 4$) for which $c_i$ is satisfied, there is a unique assignment of a boolean value to $q_i$ such that $c'_i$ is satisfied. By induction on $i$ it follows that $\psi'=\bigwedge_{i\le m} c'_i$ is a \textsc{3sat} instance such that, if $s^\ast$ is a certificate for $\psi$, then there is a unique assignment $q^\ast$ of boolean values to the $q_i$ variables that make $(s^\ast,q^\ast)$ a certificate for $\psi'$. Conversely, if $\psi$ is unsatisfiable, $\psi'$ cannot be satisfiable, otherwise every conjunction $c'_i$ could be transformed back to $c_i$ providing an easy contradiction.

This establishes a one-to-one correspondence between the certificates of YES instances of \textsc{4sat} and those of the \textsc{3sat} instances to which the latter are reduced. Now consider an instance of \textsc{a-4sat}: because of this one-to-one correspondence, it has many different certificates if and only if the corresponding \textsc{a-3sat} does, which establishes the \textbf{NP}-completeness of \textsc{a-3sat}.

In \cite{saxe80}, Saxe provides a much more general treatment than ours, and proves that the conclusion of the above reasoning actually holds for all ambiguous problem pairs where there is a one-to-one certificate correspondence. Moreover, he also generalizes the equality relation to any equivalence relation. In this paper we only deal with the EDGP and only consider the congruence equivalence relation.

\subsection{From ambiguous $3$-satisfiability to A-EDGP}
By Thm.~\ref{thm:3sat2edgp1} and the note after its proof, there is a one-to-one correspondence between \textsc{3sat} certificates and realizations of $\mathcal{G}$ up to congruence. Therefore, since the \textsc{a-3sat} is \textbf{NP}-complete and the EDGP${}_1$ is in \textbf{NP} by Prop.~\ref{prop:ndtm}, the A-EDGP${}_1$ is \textbf{NP}-complete, even if its edge weights are in $\{1,2\}$. Moreover, by Sect.~\ref{s:edgpK}, the A-EDGP${}_K$ is \textbf{NP}-hard for any $K>1$. Finally, the A-EDGP is strongly \textbf{NP}-hard.

\section{Conclusion}
\label{s:concl}
We provided a criticism of two papers by Saxe devoted to the computational complexity of the EDGP and some its variants, filling in many details that the proof sketches in the papers omitted. We concur with Saxe that the discovery of some new strongly \textbf{NP}-hard geometric problems is interesting, and plan to carry out more research in this direction: is there anything to be gained by seeing \textsc{sat} or other ``general'' problems used by humans for modelling purposes as problems to be modelled and solved geometrically? Some preliminary considerations in this sense are given in \cite{dgcomp1}.


\bibliographystyle{plain}
\bibliography{../../../dr1}

\end{document}

%% file: definecolors.tex
\definecolor{offwhite}{gray}{0.92}
\definecolor{lightgrey}{gray}{0.85}
\definecolor{lightmidgrey}{gray}{0.75}
\definecolor{midgrey}{gray}{0.6}
\definecolor{darkmidgrey}{gray}{0.5}
\definecolor{darkgrey}{gray}{0.4}
\definecolor{verydarkgrey}{gray}{0.2}
\definecolor{grey9}{gray}{0.9}
\definecolor{grey8}{gray}{0.8}
\definecolor{grey7}{gray}{0.7}
\definecolor{grey6}{gray}{0.6}
\definecolor{grey5}{gray}{0.5}
\definecolor{grey}{gray}{0.5}
\definecolor{gray}{gray}{0.5}
\definecolor{grey4}{gray}{0.4}
\definecolor{grey3}{gray}{0.3}
\definecolor{grey2}{gray}{0.2}
\definecolor{grey1}{gray}{0.1}
\definecolor{pink}{rgb}{1,0.5,0.5}
\definecolor{lightpink}{rgb}{1,0.8,0.8}
\definecolor{orange}{cmyk}{0,0.61,0.87,0}
\definecolor{lightorange}{cmyk}{0,0.58,0.95,0}
\definecolor{darkgreen}{rgb}{0,0.3,0}
\definecolor{darkred}{rgb}{0.5,0,0}
\definecolor{darkblue}{rgb}{0,0,0.5}
\definecolor{lightgreen}{rgb}{0.6,1,0.4}
\definecolor{lightcyan}{rgb}{0.8,1,0.9}
\definecolor{purple}{rgb}{0.7,0.1,0.7}
\definecolor{darkpurple}{rgb}{0.55,0.05,0.55}
\definecolor{reasonablegreen}{rgb}{0,0.5,0}
\definecolor{black}{rgb}{0,0,0}

\definecolor{header}{gray}{0.75}
\definecolor{subheader}{gray}{0.85}
\definecolor{evenrow}{gray}{0.90}
\definecolor{oddrow}{gray}{1.0}

\definecolor{er}{rgb}{1,1,0.5}
\definecolor{hdm}{rgb}{0.5,1,0.5}
\definecolor{rel}{rgb}{0.5,1,1}
\definecolor{orm}{rgb}{0.8,0.8,0.5}
\definecolor{uml}{rgb}{0.6,0.6,1}
\definecolor{www}{rgb}{1,0.5,1}
\definecolor{xml}{rgb}{0.8,0.8,1}
\definecolor{comment}{rgb}{0.75,0.75,0.75}

\definecolor{GreenYellow}{cmyk}{0.15,0,0.69,0}
\definecolor{Yellow}{cmyk}{0,0,1,0}
\definecolor{Goldenrod}{cmyk}{0,0.10,0.84,0}
\definecolor{Dandelion}{cmyk}{0,0.29,0.84,0}
\definecolor{Apricot}{cmyk}{0,0.32,0.52,0}
\definecolor{Peach}{cmyk}{0,0.50,0.70,0}
\definecolor{Melon}{cmyk}{0,0.46,0.50,0}
\definecolor{YellowOrange}{cmyk}{0,0.42,1,0}
\definecolor{Orange}{cmyk}{0,0.61,0.87,0}
\definecolor{BurntOrange}{cmyk}{0,0.51,1,0}
\definecolor{Bittersweet}{cmyk}{0,0.75,1,0.24}
\definecolor{RedOrange}{cmyk}{0,0.77,0.87,0}
\definecolor{Mahogany}{cmyk}{0,0.85,0.87,0.35}
\definecolor{Maroon}{cmyk}{0,0.87,0.68,0.32}
\definecolor{BrickRed}{cmyk}{0,0.89,0.94,0.28}
\definecolor{Red}{cmyk}{0,1,1,0}
\definecolor{OrangeRed}{cmyk}{0,1,0.50,0}
\definecolor{RubineRed}{cmyk}{0,1,0.13,0}
\definecolor{WildStrawberry}{cmyk}{0,0.96,0.39,0}
\definecolor{Salmon}{cmyk}{0,0.53,0.38,0}
\definecolor{CarnationPink}{cmyk}{0,0.63,0,0}
\definecolor{Magenta}{cmyk}{0,1,0,0}
\definecolor{VioletRed}{cmyk}{0,0.81,0,0}
\definecolor{Rhodamine}{cmyk}{0,0.82,0,0}
\definecolor{Mulberry}{cmyk}{0.34,0.90,0,0.02}
\definecolor{RedViolet}{cmyk}{0.07,0.90,0,0.34}
\definecolor{Fuchsia}{cmyk}{0.47,0.91,0,0.08}
\definecolor{Lavender}{cmyk}{0,0.48,0,0}
\definecolor{Thistle}{cmyk}{0.12,0.59,0,0}
\definecolor{Orchid}{cmyk}{0.32,0.64,0,0}
\definecolor{DarkOrchid}{cmyk}{0.40,0.80,0.20,0}
\definecolor{Purple}{cmyk}{0.45,0.86,0,0}
\definecolor{Plum}{cmyk}{0.50,1,0,0}
\definecolor{Violet}{cmyk}{0.79,0.88,0,0}
\definecolor{RoyalPurple}{cmyk}{0.75,0.90,0,0}
\definecolor{BlueViolet}{cmyk}{0.86,0.91,0,0.04}
\definecolor{Periwinkle}{cmyk}{0.57,0.55,0,0}
\definecolor{CadetBlue}{cmyk}{0.62,0.57,0.23,0}
\definecolor{CornflowerBlue}{cmyk}{0.65,0.13,0,0}
\definecolor{MidnightBlue}{cmyk}{0.98,0.13,0,0.43}
\definecolor{NavyBlue}{cmyk}{0.94,0.54,0,0}
\definecolor{RoyalBlue}{cmyk}{1,0.50,0,0}
\definecolor{Blue}{cmyk}{1,1,0,0}
\definecolor{Cerulean}{cmyk}{0.94,0.11,0,0}
\definecolor{Cyan}{cmyk}{1,0,0,0}
\definecolor{ProcessBlue}{cmyk}{0.96,0,0,0}
\definecolor{SkyBlue}{cmyk}{0.62,0,0.12,0}
\definecolor{Turquoise}{cmyk}{0.85,0,0.20,0}
\definecolor{TealBlue}{cmyk}{0.86,0,0.34,0.02}
\definecolor{Aquamarine}{cmyk}{0.82,0,0.30,0}
\definecolor{BlueGreen}{cmyk}{0.85,0,0.33,0}
\definecolor{Emerald}{cmyk}{1,0,0.50,0}
\definecolor{JungleGreen}{cmyk}{0.99,0,0.52,0}
\definecolor{SeaGreen}{cmyk}{0.69,0,0.50,0}
\definecolor{Green}{cmyk}{1,0,1,0}
\definecolor{ForestGreen}{cmyk}{0.91,0,0.88,0.12}
\definecolor{PineGreen}{cmyk}{0.92,0,0.59,0.25}
\definecolor{LimeGreen}{cmyk}{0.50,0,1,0}
\definecolor{YellowGreen}{cmyk}{0.44,0,0.74,0}
\definecolor{SpringGreen}{cmyk}{0.26,0,0.76,0}
\definecolor{OliveGreen}{cmyk}{0.64,0,0.95,0.40}
\definecolor{RawSienna}{cmyk}{0,0.72,1,0.45}
\definecolor{Sepia}{cmyk}{0,0.83,1,0.70}
\definecolor{Brown}{cmyk}{0,0.81,1,0.60}
\definecolor{Tan}{cmyk}{0.14,0.42,0.56,0}
\definecolor{Gray}{cmyk}{0,0,0,0.50}
\definecolor{Black}{cmyk}{0,0,0,1}
\definecolor{White}{cmyk}{0,0,0,0}

\definecolor{azure}{rgb}{0.94, 1.0, 1.0}